\documentclass[twocolumn,secthm]{autart}  

\usepackage{graphicx}          

\usepackage{comment}
\usepackage{amsmath}
\usepackage{amssymb}
\usepackage{mathtools}
\usepackage{mathrsfs}
\newtheorem{definition}{\bfseries Definition}
\newtheorem{theorem}{\bfseries Theorem}
\newtheorem{lemma}{\bfseries Lemma}
\newtheorem{corollary}{\bfseries Corollary}
\newtheorem{remark}{\bfseries Remark}
\newtheorem{example}{\bfseries Example}
\newtheorem{assumption}{\bfseries Assumption}
\newtheorem{problem}{\bfseries Problem}

\usepackage{xcolor}

\usepackage{natbib}
\usepackage{hypernat}

\usepackage[inline]{enumitem}
\setlist{nolistsep,noitemsep}
\usepackage{subcaption}
\usepackage{array}
\usepackage{etoolbox} \AtBeginEnvironment{equation}{%
	\setlength{\abovedisplayskip}{4pt}
	\setlength{\belowdisplayskip}{4pt}
	\setlength{\abovedisplayshortskip}{0pt}
	\setlength{\belowdisplayshortskip}{0pt}
}
\AtBeginEnvironment{align}{%
	\setlength{\abovedisplayskip}{4pt}
	\setlength{\belowdisplayskip}{4pt}
	\setlength{\abovedisplayshortskip}{0pt}
	\setlength{\belowdisplayshortskip}{0pt}
}
\AtBeginEnvironment{equation*}{%
	\setlength{\abovedisplayskip}{4pt}
	\setlength{\belowdisplayskip}{4pt}
	\setlength{\abovedisplayshortskip}{0pt}
	\setlength{\belowdisplayshortskip}{0pt}
}
\AtBeginEnvironment{align*}{%
	\setlength{\abovedisplayskip}{4pt}
	\setlength{\belowdisplayskip}{4pt}
	\setlength{\abovedisplayshortskip}{0pt}
	\setlength{\belowdisplayshortskip}{0pt}
}

\newcommand{\mat}[1]{\begin{bmatrix} #1 \end{bmatrix}}
\newcommand{\smat}[1]{\begin{bsmallmatrix} #1 \end{bsmallmatrix}}
\newcommand{\R}[1]{\mathbb{R}^{#1}}
\newcommand{\ran}[1]{{\rm rank}(#1)}
\newcommand{\im}[1]{{\rm Im}(#1)}
\renewcommand{\ker}[1]{{\rm Ker}(#1)}
\newcommand{\norm}[1]{\left\lVert#1\right\rVert}

\newcommand{\ab}[1]{(A_{#1},B_{#1})}

\renewcommand{\H}{\mathcal{H}}

\newcommand{\A}{\mathcal{A}}
\newcommand{\B}{\mathcal{B}}
\newcommand{\C}{\mathcal{C}}
\newcommand{\D}{\mathcal{D}}

\newcommand{\s}{_{_S}}
\newcommand{\ms}{\mathcal{S}}
\newcommand{\mt}{\mathcal{T}}
\renewcommand{\t}{_{_T}}

\usepackage[capitalise]{cleveref}
\crefname{alg}{Algorithm}{Algorithm}

\begin{document}
	
	\begin{frontmatter}

		\title{Incremental Data Driven Transfer Identification\thanksref{footnoteinfo}} 
		\vskip-1cm
		\thanks[footnoteinfo]{The material in this paper was partially presented at the 64th IEEE Conference on Decision and control (CDC), December 10-12, 2025,  Rio de Janeiro, Brazil.}
		
		\author[First]{N. Naveen Mukesh}\ead{nnmukesh@ee.iitb.ac.in},    
		\author[First]{Debraj Chakraborty}\ead{dc@ee.iitb.ac.in}           

		\address[First]{The authors are with the Department of Electrical Engineering, Indian Institute of Technology Bombay, Mumbai, India.}  


		\begin{abstract}                          
			We introduce a geometric method for online transfer  identification of a deterministic linear time-invariant system. At the beginning of the identification process, we assume access to abundant data from a system that is similar, though not identical, to the true system. In the early stages of data collection from the true system, the dataset generated  is still not sufficiently informative to enable precise identification. Consequently, multiple candidate models remain consistent with the observations available at that point. Our method picks, at each instant, the model closest to the similar system that is consistent with the current data. As more data are collected, the proposed model gradually moves away from the initial similar system and eventually converges to the true system when the data set grows to be informative. Numerical examples demonstrate the effectiveness of the incremental transfer identification paradigm, where identified models with minimal data are used to solve the pole placement problem. 
		\end{abstract}
		
	\end{frontmatter}
	
	\section{Introduction}
	\thispagestyle{empty}
	The first step in traditional control engineering practice is to model and identify the plant to be controlled. Usually, controlled experiments are conducted to obtain input–output data, after which model parameters are estimated through available techniques \citep{ljung1998system,van2012subspace}.
	However, in this framework, the data collection process needs to continue for long enough such that the data set grows to satisfy certain rank conditions \citep{katayama2005subspace}. Moreover, larger data sets are necessary to ensure adequate robustness against noise \citep{POURGHOLI2023109893}. Consequently, model-based controller design can only be initiated once these conditions are satisfied. In contrast, a growing body of recent research on data-driven control \citep{10266821,dorfler2023data} has bypassed the reliance on explicit models by constructing controllers directly from recorded data.
	However, even in these papers, a fully informative data set is necessary to solve control problems \citep{de2019formulas,celi2023closed}, and the requirement for lengthy data sets persists. In many practical contexts, such prolonged data collection or re-collection can be both inconvenient and costly. Typical examples include petrochemical processes, where the time constants are inherently large \citep{chiuso2007modeling}, as well as complex, nonlinear, or time-varying systems that require periodic re-identification to maintain accuracy \citep{lorenzen2019robust}. To address these challenges and accelerate the identification process, a growing body of literature \citep{xin2025learning,xin2022identifying,toso2023learning,kedia2024learning,
		chakrabarty2022optimizing,chakrabarty2023meta,richards2023control,
		zhang2023multi,ping2024parameter,ping2022multitask,huang2025expectation,du2024prediction}  advocates the use of previously recorded, abundant data from a similar system, together with data from the current (online or true) system to be identified. This combination, if done correctly, transfers the knowledge from the similar system to the true system, yielding an identified model during the early stages of online data collection, at a point where no unique model could have been obtained using only data from the online/true system. In this paper, we introduce a novel framework for such \emph{transfer identification} of deterministic systems, which guarantees both the existence and uniqueness of the identified model at each stage, and rigorously establishes the convergence of these models to the true system as additional data become available.
	
	While the literature on data-driven control is extensive (e.g., see \citep{de2019formulas,dorfler2023certainty,van2020noisy,berberich2020data,rueda2021data,baggio2021data,bisoffi2022controller}
	and the references therein), the integration of pre-collected data from a similar system with sparse online measurements has emerged recently \citep{xin2022identifying,chakrabarty2022optimizing,chakrabarty2023meta,richards2023control,xin2025learning,toso2023learning,kedia2024learning,
		zhang2023multi,ping2024parameter,ping2022multitask,huang2025expectation,du2024prediction}.  In  \citep{xin2022identifying,xin2025learning,kedia2024learning,toso2023learning},  data from a similar system was concatenated to the online data to enhance noise robustness in least squares, Bayesian estimation, and subspace identification settings. However, it is known that concatenating data from two distinct systems may not always yield consistency with any underlying model \citep{eising2024data}.
	Bayesian meta learning techniques were used for similar objectives in \citep{chakrabarty2023meta,richards2023control,chakrabarty2022optimizing}.  Nevertheless, the use of deep neural networks introduced significant difficulties in providing rigorous theoretical guarantees. 
	Data from similar systems was utilized to design LQR solutions in \citep{guo2025transfer,bajaj2025leveraging}, which however required to assume additional information such as impulse response data from the true system or known perturbation bounds on the true parameters. In \citep{ping2024parameter,ping2022multitask,huang2025expectation}, transfer identification was numerically explored in the context of regression models. 
	
	Transfer learning \citep{zhuang2020comprehensive} was systematically introduced in the control domain in \citep{du2024prediction} for training across multiple systems, though it is based on the assumption that the system parameters share a common column space.
	In \citep{8960476},  a common controller was found for all the systems that are consistent with the available (possibly non-informative) data. Due to the frequent infeasibility of this approach, \citep{li2023data} augmented the method with data from a similar system by constructing a set of systems that was both consistent with the available online data and within a pre-known distance of the similar system. 
	However, this method required an LMI-based convex feasibility search at each step, making it challenging to use in online settings. As opposed to most of the papers cited above, the geometric method of combining similar and true system data proposed in this paper, guarantees the existence and uniqueness of an identified model at each stage, does not require to perform any online optimization, nor does it require to tune weighting parameters or matrices to balance the contributions of data from the true system and auxiliary data from similar systems. This simplicity enhances computational efficiency, making the proposed approach well-suited for online implementation.
	
	In this article, we assume a certain level of confidence in the similar system. This confidence can come from earlier experience with that system or from engineering judgment. For instance, the similar system may be represented by a high-fidelity computer simulations or physics-based models built from the manufacturer’s nominal specifications. Although we anticipate that the true system differs slightly from the similar system, we rely on the similar system until sufficient evidence is obtained to suggest otherwise. To clarify this framework, we first characterize systems through the subspaces formed by the data they produce \citep{wang2024data,alsalti2024robust}, and we quantify how different two systems are by computing the distances between their corresponding subspaces  \citep{wedin2006angles,mandolesi2019grassmann}. As stated earlier, when only a short sequence of data from the true system is available up to the present time, there may exist infinitely many subspaces/systems that are consistent with the observed data. Therefore, selecting a single system from this infinite collection is not straightforward. However, if abundant data from a similar system are also available, we propose a method for constructing new subspaces that both contain the true system data up to the current instant and remain closest (according to an appropriately defined subspace distance) to the subspace associated with the similar system data. This evolving family of subspaces is then used to estimate (and repeatedly re-estimate) the system parameters at every time instant. Alternatively, these subspaces may be directly utilized to synthesize controllers through data-driven techniques. Beyond introducing this new subspace construction scheme, we prove/derive the following:
	\begin{enumerate}[label=(\roman*)]
		\item The subspaces computed by the proposed method at each time step move monotonically (according to an appropriately defined distance between subspaces) towards the unknown subspace of the true system as the amount of online data increases and satisfies the standard informativity conditions \citep{willems2005note,de2019formulas}.
		\item At every time instant, a unique state-space model can always be extracted from the constructed subspace.
		\item The model parameters obtained at each time step from the corresponding subspaces converge to the actual system parameters once the data becomes informative.
		\item The numerical implementation of the proposed paradigm and its computational complexity.
	\end{enumerate}
	Finally, with the help of numerical examples, we validate the theoretical results and demonstrate the practical relevance of the proposed approach with online/minimal data.
	
	This work extends the results presented in \citep{simsys_cdc}, where the incremental transfer identification paradigm was originally introduced. However, \citep{simsys_cdc} did not provide formal theoretical guarantees regarding how the accuracy of parameter estimates evolves as the computed subspaces incrementally converge toward the unknown true system subspace. We provide these guarantees here. This paper also includes formal proofs of several key theorems and complexity analysis of the algorithm that were omitted in \citep{simsys_cdc}. Furthermore, several novel examples are presented to illustrate the effectiveness and practical relevance of the proposed approach.

	\section{Problem Setup}
	We begin by outlining the preliminaries and formulating the problem of interest. We use standard mathematical notations from \citep{stewart1990matrix} throughout the paper.
	\subsection{True and similar system}
	\subsubsection{True System}
	Our objective is to identify the following discrete-time deterministic LTI system:%
	\begin{equation} 
		x\t(k+1)=A\t x\t(k)+B\t u\t(k), \quad k \in \{0,1,\ldots\}
		\label{true_system}
	\end{equation}%
	where  $A\t \in\mathbb{R}^{n\times n}$ and  $B\t \in\mathbb{R}^{n\times m}$ denote the unknown true system matrices to be determined. We assume that data collection begins at time $t=0$, and by time $t=i$, the input–state data of the form $\{u\t(t), x\t(t)\}_{t=0}^i$ has been recorded.
	\subsubsection{Similar System}
	Now, consider a different system having the same input and state dimensions as \eqref{true_system}:
	\begin{equation} 
		x\s(k+1)=A\s x\s(k)+B\s u\s(k), \quad k \in \{0,1,\ldots\}
		\label{sim_system}
	\end{equation}%
	The degree of similarity is specified in Assumption \ref{assump1}. We assume that input–state data of sufficient length (namely $N\s+1$) from the similar system \eqref{sim_system} is available at time $t=0$, i.e., $\{u\s(\tau), x\s(\tau)\}_{\tau=0}^{N_s}$ is already available at $t=0$.

	\subsection{Data Matrices}
	Using the available data till the $i^\text{th}$ time instant, the following  matrices are defined for the true system,
	\begin{equation}\label{eq_xu_defn}
		\begin{aligned}
			X\t^{i-}\hspace{-1mm}&:=\hspace{-1mm}\mat{x\t(0)&x\t(1)&\cdots&x\t(i -1)}\in \R{n\times i},\\
			X\t^{i+}\hspace{-1mm}&:=\hspace{-1mm}\mat{x\t(1)&x\t(2)&\cdots&x\t(i)}\in \R{n\times i}.
		\end{aligned}
	\end{equation}
	The matrix $U\t^{i-} \in \R{m\times i}$ is defined using input data similar to $X\t^{i-}$. 
	Further, we define the data matrices $X\s^{-}\in \R{n\times N\s},U\s^{-}\in \R{m\times N\s},X\s^{+}\in \R{n\times N\s}$ constructed using the available data from the similar system. 
	From \eqref{true_system} and \eqref{sim_system}, the data matrices above satisfy the following relations:
	\begin{align*}
		X\s^{+}=[A\s\;B\s]\smat{{X\s^-}\\{U\s^-}}, \quad  X\t^{i+}=[A\t\;B\t]\smat{{X\t^{i-}}\\ {U\t^{i-}}}.
	\end{align*}
	Define $M:=2n+m$, and the combined data matrix $S$ as follows: 
	\begin{equation}\label{eq_S}
		S:=\smat{{X\s^-}^\top ~{U\s^-}^\top~{X\s^+}^\top}^\top \in \R{M\times N\s}.
	\end{equation}
	Similarly, for the true system, denote the data matrix   \begin{equation}\label{eq_T_i} T_i:=\smat{{X\t^{i-}}^\top & {U\t^{i-}}^\top&{X\t^{i+}}^\top}^\top \in \R{M\times i}. \end{equation}
	Define $\Lambda:=n+m$ and $\Omega:=n+m+nm-1$.
	\begin{assumption}\label{assump1}
		The following are assumed:
		\begin{enumerate}
			\item The input $\{u\s(\tau)\}_{\tau=0}^{N\s}$ is persistently exciting \citep{willems2005note,de2019formulas} of order $n+1$. \label{asump1_s1}
			\item {$\norm{[A\t\;B\t]-[A\s\;B\s]}_F\leqslant\delta$, where $\delta$ is an unknown (small) constant, and $\norm{\cdot}_F$ denotes the  Frobenius norm.}
			\item {$\ab{\s}$ is controllable.}
			\item  {Let $\Delta A\in  \R{n\times n}$, $\Delta B \in \R{n \times m}$. Then every system in the set $\{(A\s\hspace{-1mm}+\hspace{-1mm}\Delta A,B\s\hspace{-1mm}+\hspace{-1mm}\Delta B) ~| ~\norm{[\Delta A\quad \Delta B]}_F\le\delta \}$ is controllable.} \label{asump1_s4}
			\item The current time $t=i$ satisfies $i\leqslant \Omega$.
			\item Abundant data is available from the similar system. i.e., $N_s>>\Omega$.
		\end{enumerate}
	\end{assumption}
	
	It is well established \citep{willems2005note} that, under these assumptions, ${\rm rank}\smat{{X\s^-}^\top & {U\s^-}^\top}^\top=\Lambda$. If we denote the range space of a matrix by $\im{\cdot}$ and define $\ms := \im{S}$, then $\dim(\ms) = \Lambda$. Furthermore, if we define $\mt_i := \im{T_i}$, it follows that under this setup the true system data $T_i$ may not satisfy  the above rank condition,  i.e., $\dim(\mt_i) < \Lambda$ for $i < \Omega$. On the other hand, the equality $\dim(\mt_\Omega) = \Lambda$ holds provided that a persistently exciting input 
	sequence $\{u\t(t)\}_{t=0}^\Omega$ is applied to the true system. Let ${\rm Gr}(k,N)$ denote the Grassmannian (the set of all $k$ dimensional linear subspaces of $\R{N}$) \citep{beyn2024smoothness}. Then $\ms, \mt_\Omega \in {\rm Gr}(\Lambda,M)$.

	\begin{remark}
		From \cite[Theorem 1.3]{wonhamlinear}, it is well known that if $\ab{\s}$ is controllable, then $\exists \; \varepsilon$ such that all systems in the set \linebreak $\{(A,B) | ~\norm{[A-A\s \quad  B-B\s]}_F \le\varepsilon \}$ are controllable. Hence, if $\delta\leqslant\varepsilon$, then statement \ref{asump1_s4} of Assumption \ref{assump1} follows directly.
		In addition, we note here that our approach remains valid even otherwise. However, in that case, the identified intermediate systems (say $\ab{i}$) would not necessarily be controllable. 
	\end{remark}
	
	\subsection{Problem Formulation and Preliminaries}
	\begin{problem}
		\begin{enumerate}[label=(\alph*)]
			\item Given $T_i$ at time step $i$, our goal is to identify a model $\big($say $\ab{i}\big)$ which is consistent with the data $T_i$
			and closest (in an appropriate sense) to $\ab{_S}$ by utilizing the data set $\{u\s(\tau),x\s(\tau)\}_{\tau=0}^{N_s}$.
			\item As additional data from $\ab{_T}$ becomes available,  re-identify the system at each time step to obtain a more accurate model.
		\end{enumerate}  
	\end{problem}
	\begin{figure} 
		\centering
		\hspace*{-0.28cm}
		\includegraphics[height=0.17\textwidth, width =0.53\textwidth]{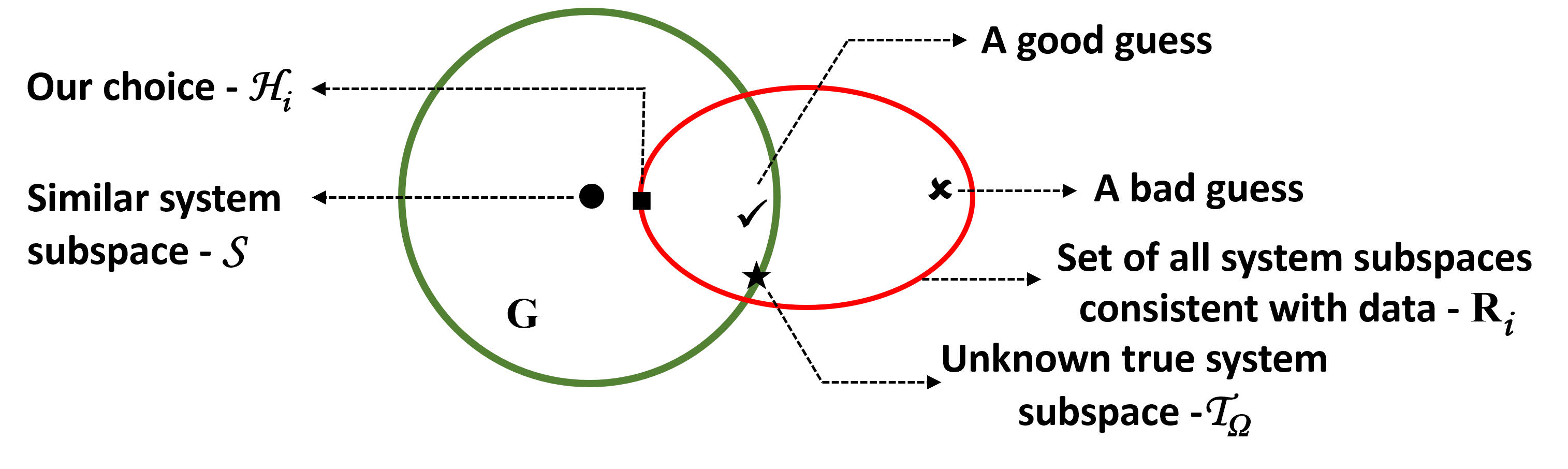}
		\caption{Illustration of  proposed paradigm}
		\label{fig_sim_sys1}
	\end{figure}
	We solve the above problem by interpreting systems as subspaces and finding the subspace (say $\H_i$) containing $\mt_i$ that is closest to, and is of the same dimension as $\ms$. We next  recall some preliminaries.
	\begin{definition}[Principal angles and vectors] \citep{bjorck1973numerical}
		Let $\mathcal{A}$, $\mathcal{B}$ denote two subspaces of $ \mathbb{R}^{M}$ with $p=\mathrm{dim}(\mathcal{A})$,  $q=\mathrm{dim}(\mathcal{B})$ and $p\geqslant q > 0$. The principal angles $0 \leqslant \theta_k \leqslant \frac{\pi}{2}$ for $k=1,2\dots,q$ is defined recursively as follows
		\begin{equation*}
			\cos \theta_k = \max_{\substack{u\in \mathcal{A} \\ v\in \mathcal{B}}}  u^\top v = u_k^\top v_k
		\end{equation*} 
		subject to $u_k^\top u_j =0$, $v_k^\top v_j =0$ $\forall ~ j = 1,2,\dots , k-1$, and $\norm{u_k}_2=1$, $\norm{v_k}_2 =1$ for all $k=1,2\dots,q$. The vectors ${u_1,u_2,\dots,u_q},{v_1,v_2,\dots,v_q}$ are called the principal vectors.
	\end{definition}
	By definition, the principal angles are such that $0 \leqslant \theta_1 \leqslant \theta_2 \leqslant \dots \leqslant \theta_q \leqslant \frac{\pi}{2}$.
	
	\begin{definition}[Distance between subspaces]
		\citep{mandolesi2019grassmann} Let  $\mathcal{A}$, $\mathcal{B}$ denote two subspaces of $ \mathbb{R}^{M}$ such that $\mathrm{dim}(\mathcal{A}) \leqslant \mathrm{dim}(\mathcal{B})$.  The distance (maximal angle) between the two subspaces is defined as, 
		\begin{equation}    \label{eq_dist_defn}
			d(\mathcal{A},\mathcal{B}) :=\max_{\substack{u\in \mathcal{A} \\ \norm{u}_{_2} =1}} \min_{\substack{v\in \mathcal{B}\\ \norm{v}_{_2} =1}} \cos^{-1} u^\top v  = \theta_{\max}
		\end{equation}
		where $ \theta_{\max}$ denotes the maximum principal angle between $\mathcal{A}$ and $\mathcal{B}$.
	\end{definition}
	The above notion of distance forms a metric space over subspaces of fixed dimensions \citep{beyn2024smoothness}.    Note that the notation $d(\cdot,\cdot)$ is used throughout the paper without proper ordering of the subspaces. The correct ordering should be understood from the context.

	\subsection{From State-Space Models to Subspace Representations}
	Our approach is illustrated in \cref{fig_sim_sys1}. By Assumption \ref{assump1}, $(A\t,B\t)\in \{ \ab{}~|\;\norm{[A\s\; B\s]-[A\; B]}_F\leqslant\delta \}$ and this set translates to a set 
	\begin{align*}
		\hspace{-1mm}\mathrm{G}:=\hspace{-1mm}\bigg\{ {\rm Im}\hspace{-0.5mm}\smat{{X^-}\\{U^-}\\X^+}~|\; X^+\hspace{-1.5mm}=[A\quad B]\smat{{X^-}\\{U^-}} \text{ with } {\rm rank}\hspace{-0.5mm}\smat{{X^-}\\{U^-}}=\Lambda \\\text{ and } \norm{[A\s\; B\s]-[A\; B]}_F\leqslant\delta  \bigg\}\subseteq{\rm Gr}(\Lambda,M)  
	\end{align*}
	This set is marked by the green circle in \cref{fig_sim_sys1}. The similar system subspace $\ms$  is denoted by a solid dot, and the unknown true system subspace $\mathcal{T}_\Omega$ is denoted by a star. When we have insufficient data (i.e., $\ran{T_i}<\Lambda$), multiple models that describe the data set exist. This set is given by $\Sigma_i=\big\{ \ab{}~|\; X\t^{i+}=[A\;B]\smat{{X\t^{i-}}\\ {U\t^{i-}}} \big\}$ and $\ab{\t}\in \Sigma_i$. Analogous to this set in ${\rm Gr}(\Lambda,M)$ we have the set $$\mathrm{R}_i:=\{ \H ~|~\H\supseteq \mt_i \text{ and } \dim(\H)=\Lambda  \}\subseteq{\rm Gr}(\Lambda,M)$$ (marked by red ellipse). Note that  $\mt_\Omega \in   \mathrm{G}\cap\mathrm{R}_i$ holds. Due to our said trust in $\ms$, we propose to choose the subspace $\H_i$ (marked by a solid square) closest to $\ms$ but lying inside the red ellipse. The method to construct such a subspace is described next.
	\begin{remark}
		We would like to emphasize the fact that the sets represented by the red ellipse and the green circle in \cref{fig_sim_sys1} are purely illustrative and are shown here to facilitate the easy interpretation of our results. In general, these sets are not constrained to such geometric forms. 
	\end{remark}

	\section{Incremental Subspaces}  \label{sec_results}
	
	We begin by describing how to construct the subspace that is closest to a given subspace (say $U$) while also containing a specified vector (say $a_1\notin U$).

	\begin{example}
		Consider a subspace $U$ and a vector $a_1 \notin U$, as shown in \cref{fig_planes1}, with  $d(U,a_1) =\theta$.  Let $b_1$ denote the projection of $a_1$ onto $U$, and define $b_2 = a_1^\perp \cap U$, where $(\cdot)^\perp$ represents the orthogonal complement. For vectors $v_1,\ldots,v_k$, we use $\langle v_1,\ldots,v_k \rangle$ to denote their linear span. It follows that among all subspaces containing $a_1$, the subspace $V=\langle a_1,b_2\rangle$, with $d(U,V)=d(U,a_1)=\theta$ is the one closest to $U$. The principal angles $\theta_i$'s together with their associated principal vectors $u_i$'s and $v_i$'s, corresponding to the subspaces $U$ and $V$, are illustrated in \cref{fig_planes1}.
		\hfill $\square$
	\end{example}

	We now employ this idea to construct subspaces by incrementally combining similar and true system data. 
	Define  $h_i:=\begin{bsmallmatrix}x(i-1)\\u(i-1)\\x(i)
	\end{bsmallmatrix}\in \R{M}$ as the vector containing the data collected from the true system $\ab{_T}$ at the $i^\text{th}$ instant. Hence,\begin{equation}\label{eq_Hibar}
		\hspace{-1mm}\mt_i = \mathrm{Im}(T_i)={\rm Im}[h_1 ~ h_2 ~ \cdots ~ h_{i}] ~\text{for}~ i=1,2,\ldots,\Omega.
	\end{equation}
	Next, we present the proposed formula for incremental subspace construction.	For $i=1,2,\ldots,\Omega$, let $\mathcal{H}_i$ be defined as follows:
	\begin{equation}\label{eq_Hi}
		\mathcal{H}_i := \mt_i + (\mt_i^\perp \cap \ms).
	\end{equation}
	Clearly, $h_i\in \mt_\Omega$ holds. Also, $\mt_i\subseteq \mt_\Omega$.       
	\begin{definition}[Partial orthogonality]\citep{mandolesi2019grassmann}
		Let $\mathcal{A}$, $\mathcal{B} ~\subseteq \R{M}$ be two subspaces. $\mathcal{A}$ is said to be partially orthogonal to $\mathcal{B}$ if $\exists$ an non-zero   vector $a \in \mathcal{A}$ such that $a^\top b =0 ~ \forall ~ b \in \mathcal{B} $. i.e $\mathcal{A} \cap \mathcal{B}^\perp \neq \{0\}$.
	\end{definition}
	Additionally, if $\mathcal{B}$ is partially orthogonal to $\mathcal{A}$, then the subspaces $\mathcal{A}$ and $\mathcal{B}$ are said to be partially orthogonal.

	In \cref{fig_planes1}, the subspaces $U$ and $V$ are partially orthogonal  as $U\cap V^\perp=U^\perp \cap V =\{0\}$ holds. Whereas, consider a subspace $V_1:=\langle a_1,b_1\rangle$, then the subspaces $U$ and $V_1$ are partially orthogonal  as $U\cap V_1^\perp=\langle u_2\rangle\neq \{0\}$ and $U^\perp \cap V_1  =U^\perp\neq \{0\}$ holds.
	
	\begin{figure} 
		\centering
		\includegraphics[height=0.15\textwidth, width =0.25\textwidth]{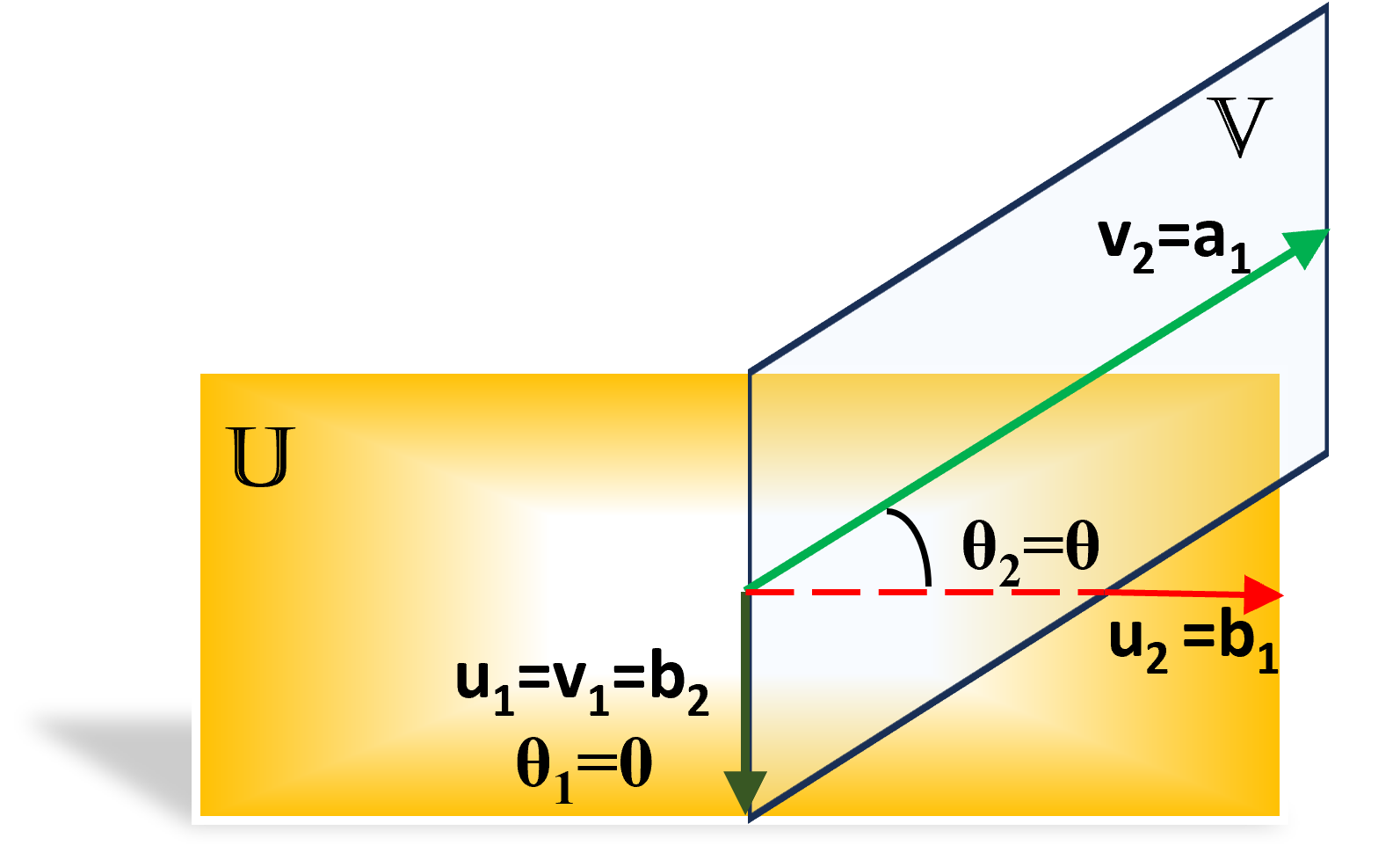}
		\caption{Closest subspace to $U$ containing the vector $a_1$}
		\label{fig_planes1}
	\end{figure}
	\begin{assumption}\label{asump_pa}
		We	assume that the subspaces $\ms$ and $\mt_\Omega$ are not partially orthogonal. (i.e., $(\ms^\perp \cap \mt_\Omega) = (\mt_\Omega^\perp \cap \ms) =\{0\}$ holds) 
	\end{assumption}
	
	\begin{remark}
		The above assumption is expected to be true even though we do not have full information or control over  the subspace $\mt_\Omega$. It is known that transversality of subspaces is a generic property (i.e., any two subspaces $\mathcal{A}$, $\mathcal{B} \subseteq \R{M}$ would have the property that $\dim(\mathcal{A}+\mathcal{B})= \min \{\dim(\mathcal{A})+\dim(\mathcal{B}),M\}$ hold generically) \cite[Section 0.16]{wonhamlinear}. Hence, in our case  any randomly chosen $\ms$, $\mt$ would satisfy $\dim(\ms+\mt_\Omega^\perp)=M (\because \dim(\ms)=\Lambda, \dim(\mt_\Omega^\perp)=M-\Lambda)$  generically. Equivalently, $\ms+\mt_\Omega^\perp=\R{M}$ and $\ms^\perp\cap\mt_\Omega=\{0\}$ hold generically. Similarly, $\ms\cap\mt_\Omega^\perp=\{0\}$ is also generically true. Hence, for any two systems $\ab{\t}$, $\ab{\s}$, the subspaces $\ms$, $\mt_\Omega$ would satisfy  Assumption \ref{asump_pa} generically and hence no prior verification is required.
	\end{remark}
	Absence of Assumption \ref{asump_pa} can lead to the following two issues:
	\begin{enumerate}[label=(\arabic*)]
		\item Necessary dimension for system identification $\Lambda$ might not be preserved.
		\item It may not be feasible to characterize a unique subspace closest (according to the metric in \eqref{eq_dist_defn}) to the similar system subspace.
	\end{enumerate}
	
	\begin{example}
		Consider the subspaces $\ms={\rm Im}\smat{1& \quad0\\0&\quad1\\1&\quad0}$, $\mt_\Omega={\rm Im}\smat{1&\quad0\\0&\quad1\\-1&\quad0}.$
		
		Note that the subspaces $\ms$ and $\mt_\Omega$ are partially orthogonal. At $t=1$, let $T_1=\smat{1&0&-1}^\top$ be available. Then, $\mt_i^\perp\cap\ms=\ms$ and $\H_i$ constructed using the formula in  \eqref{eq_Hi} yields $\H_i =\R{3}$.
		Clearly, the necessary dimension for system identification is not preserved.
		
		Next, consider the attempt to construct a subspace $\H_i$ satisfying $\dim(\H_i) = \dim(\ms)$ by selecting a subspace of dimension exactly equal to $\dim(\ms) -  \dim(\mt_i)$ from $\mt_i^\perp \cap \ms$. For example, choosing $\H_i = \langle T_1, v \rangle$ with $v \in  (\mt_i^\perp \cap \ms) $) results in infinitely many possible choices for $\H_i$. Hence, the primary objective of ensuring model uniqueness is not achieved. \hfill $\square$
	\end{example}
	We now proceed to demonstrate certain fundamental properties of the constructed subspaces $\mathcal{H}_i$.
	\subsection{Properties of $\H_i$ with fixed data size}
	
	The following result characterizes the dimension of $\H_i$'s constructed in \eqref{eq_Hi} and ensures that the necessary dimensions required for the unique identification of the system are preserved.
	\begin{theorem}[Dimension preservation] \label{thm_dim_match}
		Let $\mathcal{H}_i$ be defined as in \eqref{eq_Hi}. Then, under Assumption \ref{asump_pa}, $\dim(\mathcal{H}_i) = \Lambda ~ \forall  ~ {i} \leqslant\Omega$ holds.
	\end{theorem}\vspace*{-0.5\baselineskip}
	\begin{pf}
		Consider $\mt_i$ as defined  in \eqref{eq_Hibar}.
		\begin{align} \label{prop_eq}
			\dim (\ms \cap \mt_i^\perp) \hspace{-1mm}&= M - \dim (\ms \cap \mt_i^\perp)^\perp  (\because  \ms \cap \mt_i^\perp  \subseteq \mathbb{R}^M) \nonumber\\ 
			&= M - \dim(\ms^\perp + \mt_i) \nonumber \\
			&= M -  (\dim \ms^\perp +\dim \mt_i - \dim(\ms^\perp \cap \mt_i)) \nonumber \\ 
			&= \dim \ms -  \dim \mt_i.
		\end{align}
		The last equality follows since $\dim \ms^{\perp} = M- \dim \ms$ and $\dim(\ms^\perp \cap \mt_i) = 0$ holds. The latter one holds due to our assumption of subspaces not being partially orthogonal, i.e., $\ms^\perp \cap \mt_\Omega =\{0\}$ holds, and hence, $\ms^\perp \cap \mt_i =\{0\}$ holds as $\mt_i \subseteq \mt_\Omega$.
		Now since $\mathcal{H}_i=  \mt_i + (\ms \cap \mt_i^\perp)$, 
		\begin{align*}
			\dim \mathcal{H}_i &=  \dim \mt_i + \dim (\ms \cap \mt_i^\perp) ~ (\text{$\because$ $\mt_i \perp (\ms \cap \mt_i^\perp) $}) \\
			&= \dim \ms (\text{from \eqref{prop_eq}}). 
		\end{align*}
		The above relation holds for all $i\leqslant\Omega$. By statement \ref{asump1_s1} of Assumption \ref{assump1} it follows that $\dim(\ms)=\Lambda$ and hence $\dim(\H_i)=\Lambda$ for all $i\leqslant\Omega$. \hfill $\blacksquare$
	\end{pf}\vspace*{-0.5\baselineskip}
	It follows directly that whenever sufficient data from the actual system is available (i.e., $\ran{T_i} = \Lambda$), the subspaces $\H_i$ constructed in \eqref{eq_Hi} converge to $\mt_\Omega$.
	
	\begin{corollary}\label{coro_converge}
		Let Assumption \ref{asump_pa} hold. Then, the following are true:
		\begin{enumerate}[label=(\roman*)]
			\item If $\dim(\mt_i)=\Lambda$, then $\H_i=\mt_\Omega$.
			\item If $\dim(\mt_i)=0$, then $\H_i=\ms$.\label{coro_s2}
		\end{enumerate}	
	\end{corollary}\vspace*{-0.5\baselineskip}
	\begin{pf}
		From \eqref{prop_eq}, we get $\dim (\ms \cap \mt_i^\perp)=0$ when $\dim(\mt_i)={\Lambda}$. Therefore,  \eqref{eq_Hi} becomes $\H_i=\mt_i+0=\mt_\Omega$ (as $\mt_i=\mt_\Omega$ when $\dim(\mt_i)={\Lambda}$). Statement \ref{coro_s2} also follows similarly. \hfill $\blacksquare$
	\end{pf}\vspace*{-0.5\baselineskip}
	\begin{remark}
		From \eqref{eq_Hi} and \eqref{prop_eq}, we infer that the subspaces $\H_i$ possess an intrinsic property: at the initial stages, they rely more heavily on similar system data, since $ \dim\left(\mt_i^\perp \cap \ms\right) > \dim(\mt_i)$. However, as more system data becomes available, the reliance on similar system data decreases, because $ \dim\!\left(\mt_i^\perp \cap \ms\right) < \dim(\mt_i)$.  This constitutes a key feature of the construction: the subspaces naturally adjust their dependence on similar system data without the need to introduce explicit weight parameters \citep{xin2025learning,ping2022multitask,zhang2023multi,ping2024parameter,huang2025expectation}.
	\end{remark}
	
	Denote $\sigma_{\min}(\cdot)$ to be the minimum singular value of a matrix.
	The following lemmas will help us prove our main results. 
	\begin{lemma} \label{lem_dist_eq}
		Let $\mathcal{A}$, $\mathcal{B}{~\subseteq \R{M}}$ be two subspaces and $Q_\mathcal{A}$, $Q_\mathcal{B}$ be matrices used to represent their orthonormal basis respectively.  Let $\D \subseteq \mathcal{A}\cap \mathcal{B}^\perp$ and $Q_\D$ be a matrix containing orthonormal vectors that span $\mathcal{D}$. Define $\C=\B\oplus\D$, and let $\dim(\A)\geqslant\dim(\C)$. Further, 
		{assume that $\B$ is not partially orthogonal to $\A$.} Then, \begin{enumerate}[label=(\roman*)]
			\item $ \sigma_{\min} (Q_\mathcal{A}^\top Q_\mathcal{B})=\sigma_{\min}(Q_\mathcal{A}^\top[Q_\mathcal{B} \quad Q_\D])$ or equivalently
			$d( \mathcal{A}, \mathcal{B}) = d( \mathcal{A}, \mathcal{C})$.
			\item If $\D = \mathcal{A}\cap \mathcal{B}^\perp$, then $\dim(\C)=\dim(\A)$ and the subspace $\C$ is unique. \label{lem1_p2}
		\end{enumerate}
	\end{lemma}\vspace*{-\baselineskip}
	\begin{pf}
		See Appendix. \hfill $\blacksquare$
	\end{pf}\vspace*{-\baselineskip}
	In the following lemma, we prove the converse statement of \cref{lem_dist_eq}. 
	\begin{lemma}\label{lem_dist_eq2}
		Consider three subspaces $\A$, $\B$, $\C{~\subseteq \R{M}}$ and let  $\B\subseteq\C$.	If $d(\A,\B) =d(\A,\C)$, then $\exists$ a subspace $\mathcal{D}\subseteq \A\cap\B^\perp$ such that $\C=\B \oplus \D$  with $\dim(\D)=\dim(\A\cap\C)-\dim(\A\cap\B)$.
	\end{lemma}\vspace*{-\baselineskip}
	\begin{pf}
		See Appendix. \hfill $\blacksquare$
	\end{pf}\vspace*{-\baselineskip}
	The following theorem ensures that the subspaces $\mathcal{H}_i$'s constructed in \eqref{eq_Hi} are closest to $\ms$ among all subspaces that contain $\mt_i$ at the $i$-th instant.
	
	\begin{theorem}[Nearest to $\ms$ and consistent with data] \label{thm_Hi_near}
		Let  $\mt_i$, $\mathcal{H}_i$ be defined as in \eqref{eq_Hibar} and \eqref{eq_Hi} respectively. 	
		Let Assumption \ref{asump_pa} hold    and let $\dim(\mt_i)=i ~\forall~i\in\{1,2,\ldots,\Lambda\}$. Then, 
		\begin{enumerate}[label=(\roman*)]
			\item $d(\ms, \mathcal{H}_i) = d(\ms, \mt_i)$ \label{thm_H0i_p1}
			\item $d(\ms, \mathcal{H}_i) < d(\ms, \widetilde{\mathcal{H}}_i)$ for any subspace $\widetilde{\mathcal{H}}_i \neq \H_i$ such that $\mt_i \subseteq  \widetilde{\mathcal{H}}_i$ and $\dim(\widetilde{\mathcal{H}}_i)=\Lambda$.\label{thm_Hi_near_p1}
		\end{enumerate} 
	\end{theorem}\vspace*{-0.5\baselineskip}
	\begin{pf}
		As $\mathcal{H}_i$, $\mt_i$ are related by \eqref{eq_Hi}, \ref{thm_H0i_p1} follows from \cref{lem_dist_eq}. 
		
		From \cref{lem_dist_eq} and \cref{lem_dist_eq2}, we conclude  that for all  subspaces $\mathcal{D}\nsubseteq \A\cap\B^\perp$, we get $d(\A,\B) <d(\A,\C)$ where $\C=\B \oplus \D$. Hence, $d(\ms, \mathcal{H}_i) < d(\ms, \widetilde{\mathcal{H}}_i)$. \hfill $\blacksquare$
	\end{pf}\vspace*{-\baselineskip}
	By the above theorem, it is now clear that our choice of subspace $\H_i$ (marked by a solid square) in \cref{fig_sim_sys1} is the closest to $\ms$ and consistent with the data $T_i$ (lying inside the red ellipse).
	\begin{theorem} \label{thm_H0star_near}
		Let  $\mt_i$, $\mathcal{H}_i$ be defined as in \eqref{eq_Hibar} and \eqref{eq_Hi}, respectively. 
		Let Assumption \ref{asump_pa} hold and $\dim(\mt_i)=i ~\forall~i\in\{1,2,\ldots,\Lambda\}$. Then, 	
		\begin{enumerate}[label=(\roman*)]
			\item $d(\mt_\Omega, \mathcal{H}_i) = d(\mt_\Omega, \mt_i^\perp \cap \ms )$. \label{thm_H0star_near_p1}
			\item $d(\mt_\Omega, \mathcal{H}_i) < d(\mt_\Omega, \widetilde{\mathcal{H}}_i)$ for any subspace $\widetilde{\mathcal{H}}_i \neq \H_i$ such that $(\mt_i^\perp\cap\ms) \subseteq  \widetilde{\mathcal{H}}_i$. \label{thm_H0star_near_p2}
		\end{enumerate}
	\end{theorem}\vspace*{-0.5\baselineskip}
	\begin{pf}
		By \cref{lem_dist_eq},  it follows that $ d(\mt_\Omega, \mt_i^\perp \cap \ms )=d(\mt_\Omega, (\mt_i^\perp \cap \ms) +\mt_i )$ as $\mt_i\subseteq\mt_\Omega$ and $\mt_i \perp (\mt_i^\perp \cap \ms)$.
		
		Next to prove \ref{thm_H0star_near_p2}, note that $\H_i$ can be written as
		\begin{equation}\label{eq_thm_Hostar_near1}
			\H_i=(\mt_i^\perp \cap \ms)+((\mt_i^\perp \cap \ms)^\perp\cap\mt_\Omega),
		\end{equation}
		as $ ((\mt_i^\perp \cap \ms)^\perp\cap\mt_\Omega)=((\mt_i +\ms^\perp)\cap\mt_\Omega)=\mt_i (\because \ms^\perp\cap \mt_\Omega =\{0\} \text{ by }$ 
		Assumption \ref{asump_pa}).
		
		Comparing \eqref{eq_thm_Hostar_near1} and \eqref{eq_Hi},  then by  \cref{thm_H0star_near} statement \ref{thm_H0star_near_p2}, \ref{thm_H0star_near_p2} follows. \hfill $\blacksquare$
	\end{pf}\vspace*{-0.5\baselineskip}
	The above theorem says that $\H_i$'s constructed as in \eqref{eq_Hi} is the closest to $\mt_\Omega$ among all subspaces containing $(\mt_i^\perp \cap \ms)$.
	Next, we study the monotonic properties of the distance between $\mathcal{H}_i$ and the subspaces $\ms$ and $\mt_\Omega$  as the availability of data increases.

	\subsection{Properties of $\H_i$ under increasing data availability}
	The following theorem establishes the first monotonicity property relating the distances between the constructed subspaces $\H_i$ and the similar system subspace $\ms$.
	\begin{theorem}[Monotonic divergence from $\ms$] \label{thm_H0i}
		Let  $\mt_i$, $\mathcal{H}_i$ be defined as in \eqref{eq_Hibar} and \eqref{eq_Hi} respectively. 	
		Let Assumption \ref{asump_pa} hold and let $\dim(\mt_i)=i ~\forall~i\in\{1,2,\ldots,\Lambda\}$. Then, 
		\begin{enumerate}[label=(\roman*)]
			\item  $d(\ms, \mathcal{H}_i) < d(\ms, \mathcal{H}_{j}) ~\forall ~\mathcal{H}_i\neq\mathcal{H}_{j}$ for $i<j\leqslant \Lambda$. \label{thm_H0i_p3}
			\item $d(\ms, \mathcal{H}_i) < d(\ms, \mathcal{T}_{\Omega}) ~\forall~i\in\{1,2,\ldots,\Lambda\}$ such that $\H_i\neq\mt_\Omega$.
		\end{enumerate} 
	\end{theorem}\vspace*{-0.5\baselineskip}
	\begin{pf}
		To prove \ref{thm_H0i_p3}, we compare $\H_{j}={\mt_j}+ (\mt_j^\perp \cap \ms)$ and $\mathcal{H}_i = \mt_i + (\mt_i^\perp \cap \ms) $ where $i<j$. Under the assumption $\dim(\mt_i)=i$,   $\mt_{j}\supset\mt_{i}$ and $(\mt_j^\perp \cap \ms)\subset (\mt_i^\perp \cap \ms)$ holds, where the latter inclusion holds because $\mt_j^\perp\subset \mt_i^\perp$ and $\dim(\mt_j^\perp \cap \ms)=\Lambda-j<\dim(\mt_i^\perp \cap \ms)=\Lambda-i$ for $i<j$. Then $\exists$ a subspace $\mathcal{F}$ of dimension $j-i$ such that $(\mt_i^\perp \cap \ms)=(\mt_j^\perp \cap \ms)+\mathcal{F}$. Hence, $\H_i$ can be rewritten as $\H_i=\mt_i+(\mt_j^\perp \cap \ms)+\mathcal{F}$. Similarly, \begin{align}\H_j=\mt_i+\langle h_{i+1},\ldots,h_j\rangle +(\mt_j^\perp \cap \ms),\label{eq_pf_thm4}\end{align} 
		where $\mt_j=\mt_i+\langle h_{i+1},\ldots,h_j\rangle$ due to \eqref{eq_Hibar}.
		It is easy to show that, when $\mathcal{H}_i\neq\mathcal{H}_{j}$, $\mathcal{F}\neq\langle h_{i+1},\ldots,h_j\rangle$. 
		\begin{align*}
			\hspace{-1.5mm}d(\ms, \mathcal{H}_i)\stackrel{(a)}{=}	d(\ms, \mt_i) &\stackrel{(b)}{=} d(\ms, \mt_i+(\mt_j^\perp \cap \ms) )\\
			&< d(\ms, \mt_i+(\mt_j^\perp \cap \ms)+ \langle h_{i+1},\ldots,h_j\rangle)\\&=d(\ms, \mathcal{H}_j) (\text{from }\eqref{eq_pf_thm4}) 
		\end{align*} Equality $(a)$ is due to \cref{thm_Hi_near} statement \ref{thm_H0i_p1} and equality  $(b)$ holds due to \cref{lem_dist_eq} as $(\mt_j^\perp \cap \ms) \subset(\mt_i^\perp \cap \ms)$.  
		The inequality holds due to \cref{lem_dist_eq2} as $\langle h_{i+1},\ldots,h_j\rangle\nsubseteq(\mt_i^\perp \cap \ms)$.\\
		$d(\ms, \mathcal{H}_i) < d(\ms, \mathcal{T}_{\Omega})$ also follows in a similar way. \hfill $\blacksquare$
	\end{pf}\vspace*{-0.5\baselineskip}
	As a direct consequence of the above theorem, we obtain,
	$$\underbrace{d(\ms, \mathcal{H}_{0})}_{=0}<d(\ms,\H_1)<\cdots<d(\ms,\H_{\Lambda-1})<\underbrace{d(\ms,\H_\Lambda)}_{=d(\ms, \mathcal{T}_{\Omega})}$$
	assuming that a new direction of $\mt_\Omega$ is obtained at each stage. The above theorem implies that, as additional data become available, the proposed model sequentially diverges from the similar system.

	In the following theorem, it will be shown that the proposed model shifts closer to the true system subspace sequentially.

	\begin{theorem}[Monotonic convergence to $\mt_\Omega$] \label{thm_H0star}
		Let  $\mt_i$, $\mathcal{H}_i$ be defined as in \eqref{eq_Hibar} and \eqref{eq_Hi}, respectively. 
		Let Assumption \ref{asump_pa} hold and $\dim(\mt_i)=i ~\forall~i\in\{1,2,\ldots,\Lambda\}$. Then, 	
		\begin{enumerate}[label=(\roman*)]
			\item $d(\mt_\Omega, \mathcal{H}_i) > d(\mt_\Omega, \mathcal{H}_{j}) ~\forall ~\mathcal{H}_i\neq\mathcal{H}_{j}$ for $i<j\leqslant \Lambda$.\label{thm_H0star_p2}
			\item $d(\mt_\Omega, \ms) > d(\mt_\Omega, \mathcal{H}_i) ~\forall~i\in\{1,2,\ldots,\Lambda\}$ such that $\H_i\neq\mt_\Omega$.\label{thm_H0star_p3}
		\end{enumerate}
	\end{theorem}\vspace*{-0.5\baselineskip}
	\begin{pf}
		As shown in the proof of Theorem \ref{thm_H0i}, we can write $\H_i=\mt_i+(\mt_j^\perp \cap \ms)+\mathcal{F}$, $\H_j=\mt_i+\langle h_{i+1},\ldots,h_j\rangle +(\mt_i^\perp \cap \ms)$ for $i<j$ and evidently, when $\mathcal{H}_i\neq\mathcal{H}_{j}$, $\mathcal{F}\neq\langle h_{i+1},\ldots,h_j\rangle$.
		Then,%
		\begin{align}\label{eq_thm_H0star3}
			\hspace*{-4mm}    \begin{aligned}
				d(\mt_\Omega, \mathcal{H}_{j})\hspace{-0.5mm}&\stackrel{(a)}{=}\hspace{-0.5mm}d(\mt_\Omega, \mt_j^\perp \cap \ms ) \hspace{-0.5mm}\stackrel{(b)}{=}\hspace{-0.5mm} d(\mt_\Omega, (\mt_j^\perp \cap \ms) + \mt_i) \\ 
				&< d(\mt_\Omega, (\mt_j^\perp \cap \ms) + \mt_i + \mathcal{F})\\&=d(\mt_\Omega, \mathcal{H}_i).
			\end{aligned}
		\end{align}%
		Here $(\mt_j^\perp \cap \ms )^\perp = \mt_j +\ms^\perp$ and hence, $(\mt_j^\perp \cap \ms )^\perp\cap \mt_\Omega = \mt_j$ as $\ms^\perp \cap \mt_\Omega=\{0\}$ due to Assumption \ref{asump_pa}. Equality $(a)$ in \eqref{eq_thm_H0star3}  is due to \cref{thm_H0star_near} statement \ref{thm_H0star_near_p1} and equality  $(b)$
		holds due to \cref{lem_dist_eq} as ${\mt_i \subset\mt_j}=(\mt_j^\perp \cap \ms )^\perp\cap \mt_\Omega$.
		The inequality in \eqref{eq_thm_H0star3} follows due to \cref{lem_dist_eq2}, as $\mathcal{F}\nsubseteq {\mt_j}$. The last equality in \eqref{eq_thm_H0star3} follows since $\H_i=\mt_i+(\mt_j^\perp \cap \ms)+\mathcal{F}$.
		
		One can show $d(\mt_\Omega, \ms) > d(\mt_\Omega, \mathcal{H}_i)$  using similar arguments and hence is omitted. \hfill $\blacksquare$
	\end{pf}\vspace*{-0.5\baselineskip}
	\begin{figure} \centering
		\includegraphics[height=0.16\textwidth, width =0.28\textwidth]{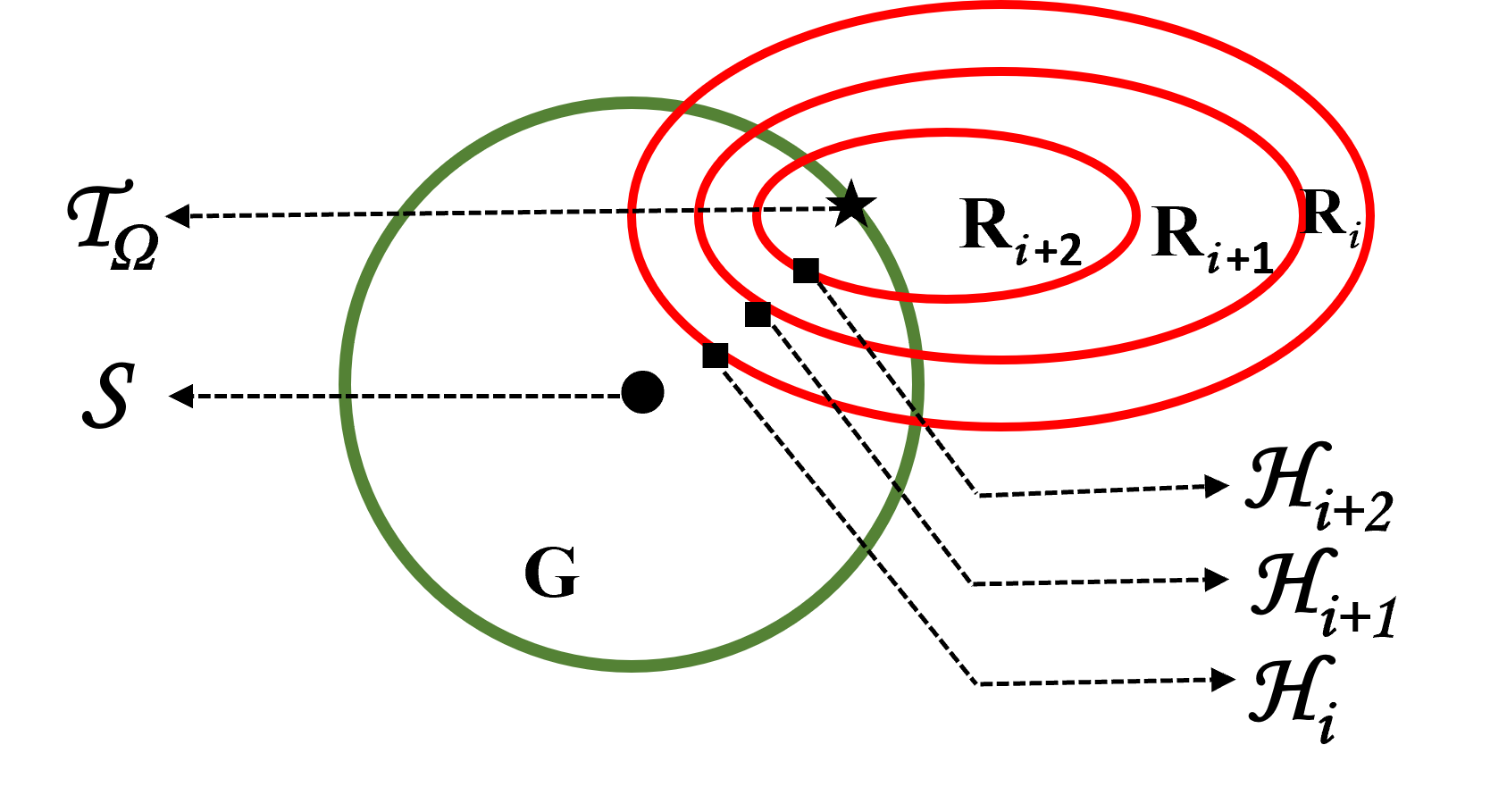}
		\caption{Variation of $d(\ms,\H_i)$ and $d(\mt_\Omega,\H_i)$}
		\label{fig_sim_sys2}
	\end{figure}
	As a consequence of above theorem, we have, \begin{align*} \hspace{-1mm}d(\mt_\Omega, \ms)\hspace{-1mm}>d(\mt_\Omega, \mathcal{H}_1)\hspace{-1mm}>\cdots>d(\mt_\Omega, \mathcal{H}_{\Lambda-1})\hspace{-1mm}>d(\mt_\Omega, \mathcal{H}_{\Lambda})=0, \end{align*} 
	assuming that we get a new direction of $\mt_\Omega$ at each stage. 
	\cref{fig_sim_sys2}  illustrates the variation of  $d(\ms,\H_i)$ and $d(\mt_\Omega,\H_i)$ with increasing $i$. 
	Note that, as $i$ increases, $d(\ms,\H_i)$ increases according to \cref{thm_H0i} and $d(\mt_\Omega,\H_i)$ decreases according to \cref{thm_H0star}. We emphasize that the red ellipse and the green circle in \cref{fig_sim_sys2} are purely illustrative in nature, included only to support the clear interpretation of our result.
	
	\subsection{Special cases arising from the proposed framework}
	Next, we see some special cases of our setup.  The following result deals with a special case when $\H_i$ constructed matches with $\mt_\Omega$ at some intermediate stage when we have insufficient data.
	\begin{lemma}\label{lem_eq}
		Let $\dim(\mt_i)=i ~\forall~i\in\{1,2,\ldots,\Lambda\}$, let Assumption \ref{asump_pa} hold and $\mt_i\neq\mt_\Omega$. Then, the following are equivalent
		\begin{enumerate}[label=(\arabic*)]
			\item $\mt_i^\perp \cap \ms=\mt_i^\perp \cap \mt_\Omega$.\label{lem_eq_p1}
			\item $\H_i=\mt_\Omega$.  \label{lem_eq_p2}
			\item $d(\H_i,\ms)=d(\mt_\Omega,\ms)$.\label{lem_eq_p3}
			\item $d(\mt_i,\ms)=d(\mt_\Omega,\ms)$.\label{lem_eq_p4}
		\end{enumerate}
	\end{lemma}\vspace*{-0.5\baselineskip}
	\begin{pf}
		\ref{lem_eq_p1} $\Rightarrow$ \ref{lem_eq_p2} When $\mt_i^\perp \cap \ms=\mt_i^\perp \cap \mt_\Omega$, clearly we have $\H_i \subseteq \mt_\Omega$ as $(\mt_i^\perp \cap \ms) \subseteq \mt_\Omega$ holds in addition to $\mt_i \subseteq \mt_\Omega$. Also as $\dim(\H_i)=\dim(\mt_\Omega)$, we get $\H_i=\mt_\Omega$.\\
		\ref{lem_eq_p2} $\Rightarrow$ \ref{lem_eq_p3}  follows directly.\\
		\ref{lem_eq_p3} $\Rightarrow$ \ref{lem_eq_p4} follows from \cref{thm_Hi_near}, statement \ref{thm_H0i_p1}.\\
		\ref{lem_eq_p4} $\Rightarrow$ \ref{lem_eq_p1} 
		By \cref{lem_dist_eq2}, $\exists$ a subspace $\mathcal{F}$ such that $\mt_\Omega=\mt_i+\mathcal{F}$ such that $\mathcal{F}\subseteq(\ms\cap\mt_i^\perp)$. But $\dim(\mathcal{F})=\dim(\ms\cap\mt_{i}^\perp)$ and hence $\mathcal{F}=(\ms\cap\mt_i^\perp)$. Therefore, $\mt_\Omega=\mt_i+(\ms\cap\mt_i^\perp)$. Now, $(\mt_i^\perp\cap\mt_\Omega)=(\mt_i^\perp\cap\ms)$ as $(\mt_i^\perp\cap\mt_i)=\{0\}$. \hfill $\blacksquare$
	\end{pf}\vspace*{-0.5\baselineskip}
	The illustration of the above scenario is shown in \cref{fig_spl_case1}.
	In that case, all subsequently constructed $\H_i$'s will be identical and equal to $\mt_\Omega$.
	The following theorem guarantees that.
	\begin{theorem}
		Let $\dim(\mt_i)=i ~\forall~i\in\{1,2,\ldots,\Lambda\}$ and  let Assumption \ref{asump_pa} hold. For some $\mt_i$,  if the conditions in \cref{lem_eq} holds, then 
		$\H_j=\mt_\Omega ~\forall~j$ satisfying $i\leqslant j\leqslant \Lambda$.
	\end{theorem}\vspace*{-0.5\baselineskip}
	\begin{pf}
		If the conditions in \cref{lem_eq} holds, then $\H_j = \mt_j + (\mt_j^\perp\cap \ms)\subseteq\mt_\Omega$ as $(\mt_j^\perp\cap \ms)\subseteq(\mt_i^\perp\cap \ms)=(\mt_i^\perp\cap \mt_\Omega)\subseteq\mt_\Omega$ holds for $i\leqslant j$ and $\mt_j\subseteq\mt_\Omega$ holds by definition. But $\dim(\H_j)=\dim(\mt_\Omega)$ and hence $\H_j=\mt_\Omega$.~\hfill $\blacksquare$
	\end{pf}\vspace*{-0.5\baselineskip}
	\begin{figure} \centering
		\hspace*{-0.3cm}
		\begin{subfigure}[b]{0.25\textwidth}
			\includegraphics[height=0.6\textwidth, width =1.3\textwidth]{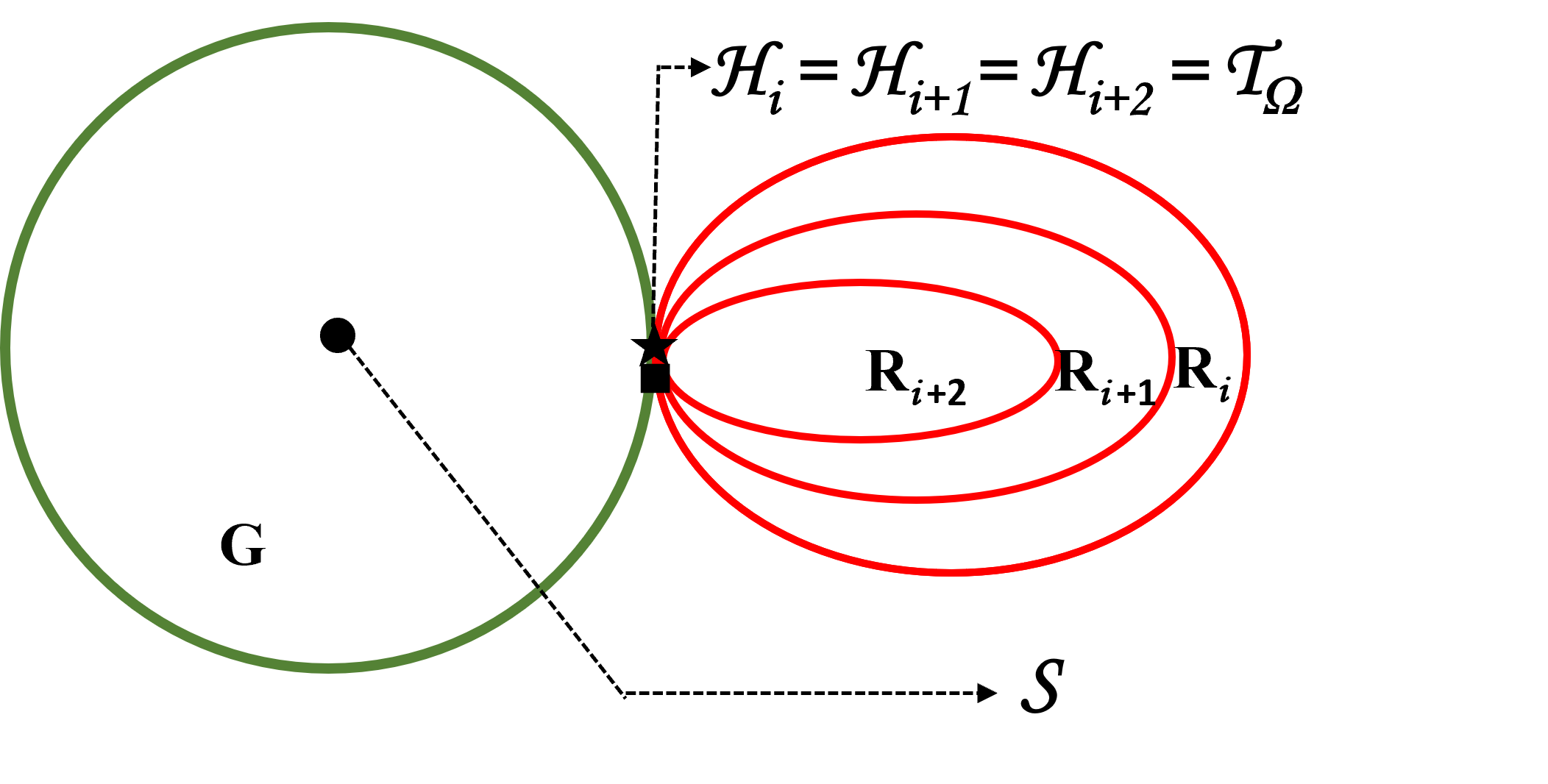}
			\caption{}
			\label{fig_spl_case1}
		\end{subfigure}\hspace*{0.2cm}		
		\begin{subfigure}[b]{0.25\textwidth}
			\includegraphics[height=0.58\textwidth, width =0.95\textwidth]{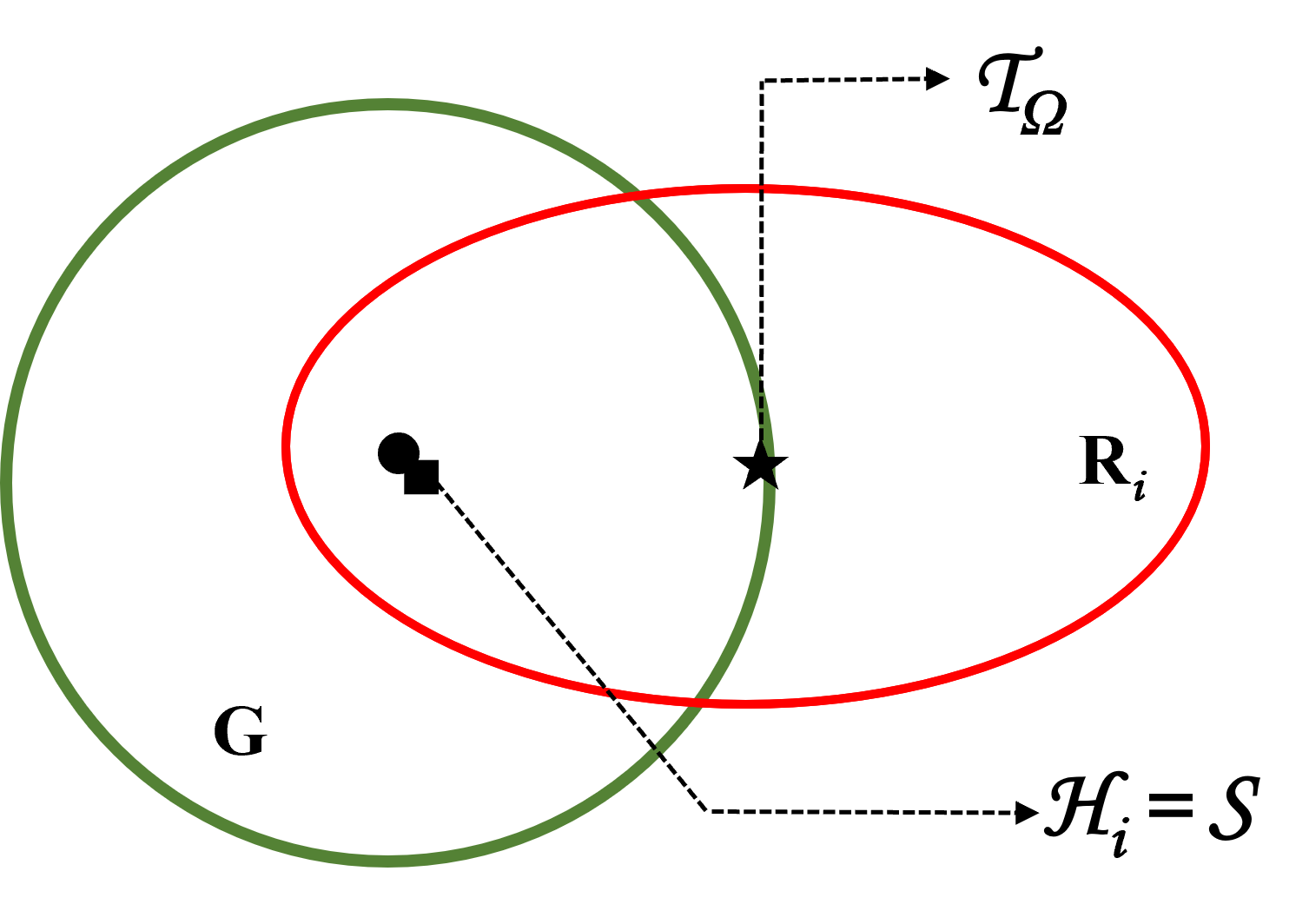}
			\caption{}
			\label{fig_spl_case2}
		\end{subfigure}
		\caption{Special cases}
	\end{figure}
	
	The following theorem  deals with another special case when the $\H_i$ constructed at the $i^\text{th}$ time instant $(i<\Omega)$ matches with $\ms$.
	\begin{theorem} 
		Let Assumption \ref{asump_pa} hold.   If $\mt_i \subseteq  \ms$, then $\H_i = \ms$.
	\end{theorem}\vspace*{-0.5\baselineskip}
	\begin{pf}
		When $\mt_i \subseteq  \ms$, clearly we have $\H_i \subseteq \ms$ as $(\mt_i^\perp \cap \ms) \subseteq \ms$ also holds. Also as $\dim(\H_i)=\dim(\ms)$, we get $\H_i=\ms$. \hfill $\blacksquare$
	\end{pf}\vspace*{-0.5\baselineskip}
	The inclusion $\mt_i \subseteq  \ms$ means that the similar system is also consistent with the dataset $T_i$. In this case, the red ellipse  contains the solid dot. (See  \cref{fig_spl_case2}) Hence, the solid dot and solid square coincide in this case (i.e., $\H_i=\ms$).

	\section{From Subspace Representations to State-Space Models}
	
	In this section, we focus on identifying a state-space model from the subspace $\mathcal{H}_i$ constructed in accordance with \eqref{eq_Hi}, and establish the existence and uniqueness of the identified model.
	
	\subsection{Existence and uniqueness of $\ab{i}$}
	
	Recall that, $\H_i$'s are created in \eqref{eq_Hi} as a sum of two separate subspaces  $\mt_i$ and ${(\mt_i^\perp\cap\ms)}$. However, if we add two arbitrary subspaces in general, then while the combined subspace might have adequate dimension, no system might exist that can produce data consistent with the combined subspace.
	
		\begin{example}
			Consider the system  $x(k+1)=ax(k)+bu(k)$ where $a,b \in \R{}$. Suppose the data set $T_1 = \smat{x(0)\\u(0)\\x(1)} = \begin{bsmallmatrix} 1\\1\\2 \end{bsmallmatrix}$ is available.  Since ${\rm rank}(T_1)=1<2$ we cannot uniquely identify $(a,b)$ from just $T_1$. Suppose we randomly complete $T_1$ to $\widetilde{H}_1=\smat{1&\quad 1\\1&\quad 1\\2&\quad 1}.$ Clearly no $(a,b)$ pairs exist which are consistent with  $a+b=2$ and $a+b=1$ simultaneously. \hfill $\square$
		\end{example}
		
		The next theorem guarantees that we are assured of a unique system $\big($denoted by $\ab{i}\big)$ consistent with any basis  $H_i$  of $\H_i$  constructed according to \eqref{eq_Hi}. Define any $H_i\in\R{M\times \Lambda}$ such that $\im{H_i}=\H_i$. Further partition 
		\begin{equation}\label{eq_part_Hi}
			H_i=\begin{bsmallmatrix}
				{X_i^-}^\top & {U_i^-}^\top & {X_i^+}^\top
			\end{bsmallmatrix}^\top 
		\end{equation}
		such that  $X_i^- \in \R{n\times \Lambda}$, $U_i^-\in \R{m\times \Lambda}$, $X_i^+\in\R{n\times \Lambda}$.
		Then, the following result holds,
		\begin{theorem}[Existence and uniqueness of $\ab{i}$]  \label{thm_ID}
			Let $\dim(\mt_i)=i ~\forall~i\in\{1,2,\ldots,\Lambda\}$.	Then, there exists an unique $\ab{i}$ pair that satisfies $X_i^+=A_iX_i^-+B_iU_i^-$,	 where $X_i^-$, $U_i^-$, $X_i^+$ are defined as in \eqref{eq_part_Hi}.
		\end{theorem}\vspace*{-0.5\baselineskip}
		\begin{pf}
			First, we prove the theorem for a specific choice of $H_i$.
			Consider	$T_i \in \R{M\times i}$ as defined in \eqref{eq_T_i}. Let $S_i \in \R{M\times (\Lambda-i)}$ be  such that its columns forms a basis for $\mt_i^\perp \cap \ms$. Partition the matrix $S_i$ as $S_i=\begin{bsmallmatrix} \widetilde{X}_i^{-\top} & \widetilde{U}_i^{-^\top} & \widetilde{X}_i^{+^\top}\end{bsmallmatrix}^\top$
			where       
			$\widetilde{X}_i^-\in \R{n\times (\Lambda-i)}$, $ \widetilde{U}_i^-\in \R{m\times (\Lambda-i)}$, $ \widetilde{X}_i^+\in \R{n\times (\Lambda-i)}$. Clearly, any such partition automatically satisfy the relations	$X\t^{i+}= \begin{bsmallmatrix}A\t & B\t\end{bsmallmatrix}\begin{bsmallmatrix}X\t^{i-}\\ U\t^{i-}\end{bsmallmatrix}$, and $\widetilde{X}_i^+=\begin{bsmallmatrix} A\s & B\s \end{bsmallmatrix}\begin{bsmallmatrix}\widetilde{X}_i^-\\ \widetilde{U}_i^- \end{bsmallmatrix}$. {Denote ${\rm Rowsp}(\cdot)$ to denote the row span of a matrix.}
			Therefore, ${\rm Rowsp}(X\t^{i+}) \subseteq {\rm Rowsp} \begin{bsmallmatrix}X\t^{i-}\\ U\t^{i-}\end{bsmallmatrix} $, ${\rm Rowsp}(\widetilde{X}_i^+) \subseteq {\rm Rowsp} \begin{bsmallmatrix}\widetilde{X}_i^-\\ \widetilde{U}_i^-\end{bsmallmatrix} $ holds.
			Since $\ran{T_i}=i$ and $\ran{S_i}=\Lambda-i$, it follows that ${\rm rank}\begin{bsmallmatrix}X\t^{i-}\\ U\t^{i-}\end{bsmallmatrix}=i$ and ${\rm rank}\begin{bsmallmatrix}\widetilde{X}_i^-\\ \widetilde{U}_i^-\end{bsmallmatrix}=\Lambda-i$.\\
			Next define $H_{oi}=\left[T_i ~S_i\right]$. Clearly ${\rm Im}\left[T_i ~S_i\right]=\H_i$ from \eqref{eq_Hi}. 
			
			{\em Claim:} ${\rm Im}	\begin{bsmallmatrix}	X\t^{i-}\\ U\t^{i-} 	\end{bsmallmatrix} \cap {\rm Im}\begin{bsmallmatrix}  \widetilde{X}_i^-\\ \widetilde{U}_i^-\end{bsmallmatrix}=\{0\}$.\\
			We prove the above claim by contradiction. Suppose $\exists ~ x\in {\R{\Lambda}}$ such that $x \in {\rm Im}	\begin{bsmallmatrix}	X\t^{i-}\\ U\t^{i-} 	\end{bsmallmatrix} \cap {\rm Im}\begin{bsmallmatrix}  \widetilde{X}_i^-\\ \widetilde{U}_i^-\end{bsmallmatrix}$ and $x\neq 0$, then $\begin{bsmallmatrix}	x\\x\t \end{bsmallmatrix} \in \mt_i$, $\begin{bsmallmatrix}	x\\x\s \end{bsmallmatrix} \in (\mt_i^\perp \cap \ms)$ where $x\t:= \smat{A\t & B\t}x$, $x\s:= \smat{A\s & B\s}x$.\\
			But, $\smat{x^\top &x\t^\top}\hspace{-1mm}\smat{x\\x\s}\hspace{-1mm}=\hspace{-1mm}0$ as $\mt_i \hspace{-1mm}\perp \hspace{-1mm}(\mt_i^\perp \cap \ms)$. Hence, 
			$x^\top x+x^\top\smat{A\t & B\t}^\top\smat{A\s & B\s}x=0$.\\ i.e., $x^\top \left(  {\rm I}_\Lambda+\smat{A\t & B\t}^\top\smat{A\s & B\s}\right)x=0$   must be true for any $\ab{_T}$, $\ab{_S}$ pair. Therefore, $x=0$. Contradiction.\\
			Hence, ${\rm rank}\begin{bsmallmatrix} X\t^{i-} &\quad \widetilde{X}_i^-\\ U\t^{i-}  &\quad \widetilde{U}_i^-\end{bsmallmatrix}={\Lambda}$.
			Also as $\ran{H_{0i}}=\Lambda$, ${\rm Rowsp}\hspace{-1mm}\left[X\t^{i+} ~ \widetilde{X}_i^+\right] \hspace{-1mm}\subseteq\hspace{-1mm} {\rm Rowsp}\begin{bsmallmatrix} X\t^{i-} &\quad \widetilde{X}_i^-\\ U\t^{i-}  &\quad \widetilde{U}_i^-\end{bsmallmatrix}$ is true. Hence, an unique $\ab{i}$ satisfying $\left[X\t^{i+} ~ \widetilde{X}_i^+\right]=[A_i~ B_i] \begin{bsmallmatrix} X\t^{i-} &\quad \widetilde{X}_i^-\\ U\t^{i-}  &\quad \widetilde{U}_i^-\end{bsmallmatrix}$ exists.
			
			It is now easy to argue that any matrix satisfying $\im{H_i}=\H_{i}$ partitioned as in \eqref{eq_part_Hi} is such that ${\rm rank}\smat{{X_i^-} \\ {U_i^-}}=\Lambda$ and an unique $\ab{i}$ exists. \hfill $\blacksquare$
		\end{pf}  \vspace*{-0.5\baselineskip}

		By the above theorem, the unique system consistent with the data in \eqref{eq_part_Hi} is determined as \citep{de2019formulas}:  
		\begin{equation}\label{eq_ABi}
			\smat{A_i \quad B_i}=X_i^+\smat{X_i^-\\ U_i^-}^\dagger.\end{equation}
		
		{\begin{example}\label{eg_scalar1}
				Let the data matrices $ S=\smat{1&\quad0\\0&\quad1\\0.7&\quad0.7}, T_1=\smat{1\\1\\1}$
				be available from a system with $m=n=1$. From the data set $S$, the similar system can be uniquely identified as $\ab{\s}=(0.7,0.7)$. Whereas, corresponding to $T_1$, the set of all consistent systems with the data is given by $\Sigma_1=\{(a,b)| a+b=1 \text{ where } a,b\in \R{}\}$. Note that $\Sigma_1$  is an unbounded set and is denoted by the red line in \cref{fig_eg_ab}.   The unknown true system $\ab{\t}$ can lie anywhere on this line. Here, $\mt_i^\perp\cap\ms=\langle\smat{1\\-1\\0}\rangle$ (shown by blue line). Hence, $\H_i=\langle\smat{1\\1\\1},\smat{1\\-1\\0}\rangle$
				and the corresponding system is $\ab{i}=(0.5,0.5)$. \hfill $\square$
		\end{example}}

		\begin{figure} 
			\centering
			\hspace*{0.3cm}
			\includegraphics[height=0.16\textwidth, width =0.3\textwidth]{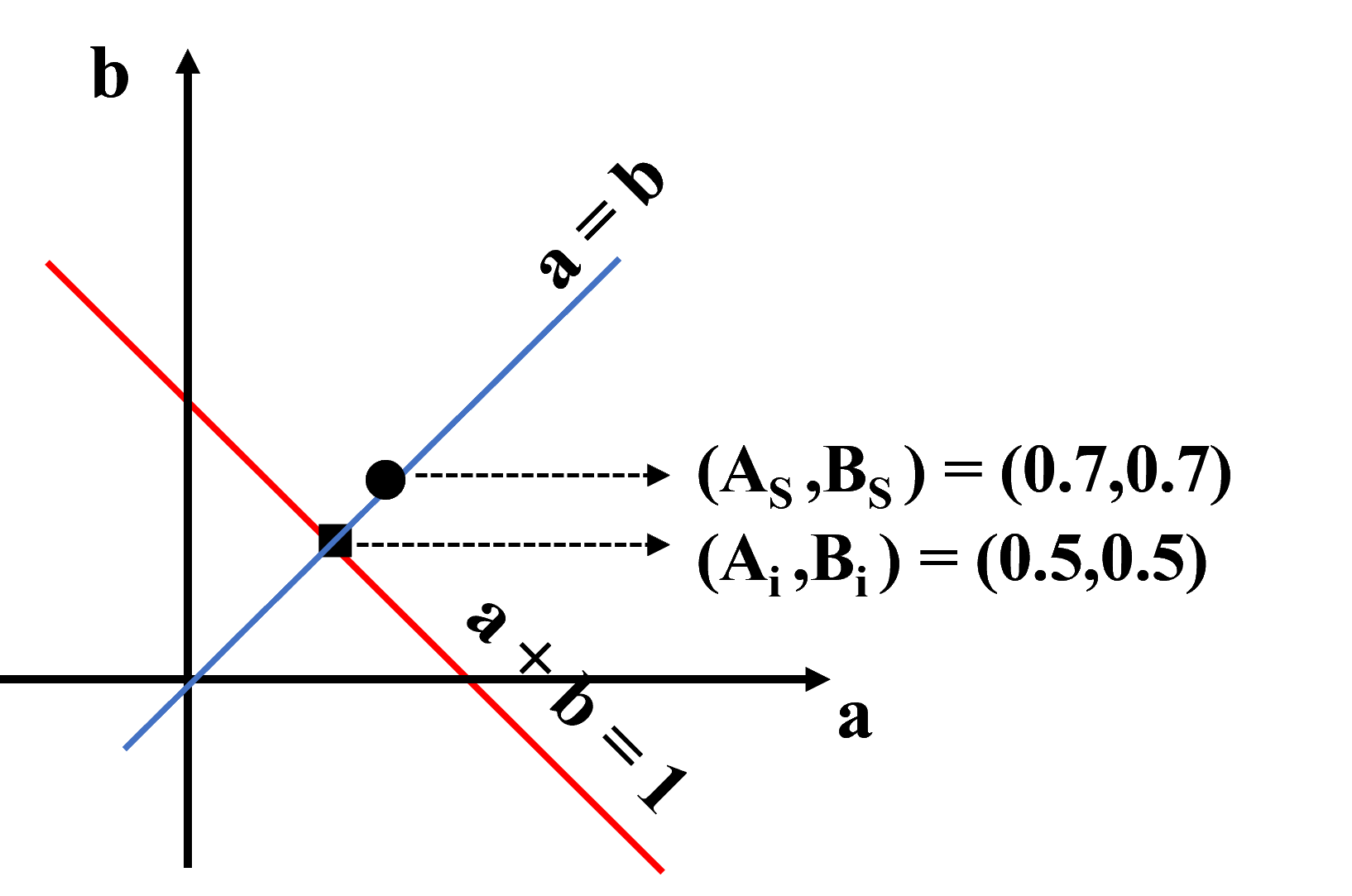}
			\caption{Graphical illustration of Example \ref{eg_scalar1}}
			\label{fig_eg_ab}
		\end{figure}
		
		The following illustrates the fact that appending data from the similar system in a way that guarantees the existence of a system might result in systems $\ab{i}$ that are far from the similar system, the true system.
		\begin{example}
			Consider the same data matrices $S$, $T_1$ as in  Example \ref{eg_scalar1}. Suppose we randomly complete $T_1$ to $\widetilde{H}_1=\smat{1&\quad 1.09302\\1&\quad 1\\1&\quad 1.4651}$ where $\smat{1.09302\\1\\1.4651} \in {\rm Im}(S)$. Also, here ${\rm Rowsp}\smat{1&\quad 1.4651}\subseteq{\rm Rowsp}\smat{1&\quad 1.09302\\1&\quad 1}$ and $\ran{\widetilde{H}_1}=\Lambda$ and hence a unique system corresponding to $\widetilde{H}_1$ exists and the unique system is $(\widetilde{A}_1,\widetilde{B}_1)=(5,-4)$. But note that $(\widetilde{A}_1,\widetilde{B}_1)$ is far from both $\ab{\s}$ and $\ab{\t}$. Hence, we can infer that combining data from similar systems in an arbitrary way would not be beneficial. \hfill $\square$
		\end{example}

		To summarize, the subspaces \(\mathcal{H}_i\), constructed with respect to the metric defined in \eqref{eq_dist_defn}, satisfy the following properties:
		\begin{enumerate*}[label=(\roman*)]
			\item {Dimension preservation} (\cref{thm_dim_match})
			\item {Minimizes deviation from the similar system}  (\cref{thm_Hi_near})
			\item {Monotonic divergence from similar system}  (\cref{thm_H0i}, statement \ref{thm_H0i_p3})
			\item {Monotonic convergence to true system}  (\cref{thm_H0star}, statement~\ref{thm_H0star_p2})
			\item Each constructed $\mathcal{H}_i$ lies closer to $\mt_\Omega$ than $\ms$ does to $\mt_\Omega$.
			(\cref{thm_H0star}, statement~\ref{thm_H0star_p3})
			\item  \(\mathcal{H}_i\)'s converge to \(\mt_\Omega\). (\cref{coro_converge})
			\item {Intermediate system identification} (\cref{thm_ID}).
		\end{enumerate*}
		\subsection{Identification of Autonomous systems}
		One of the main assumptions of our work is that we have a fully informative dataset from a similar system available. i.e., $\dim(\ms)=\Lambda$. But, for autonomous systems, since the input component is zero, we would have $\dim(\ms) < \Lambda$. However, with slight modifications to the constructed data matrices, our approach can also be used for incrementally identifying autonomous systems. 
		Next, we will show how our approach can be applied for autonomous systems. Consider 
		\begin{align*} 
			\begin{aligned}
				&\text{True system: } \hskip-2mm &&x\t(k+1)=A\t x\t(k), \quad k \in \{0,1,\ldots\}\\ 
				&\text{Similar system: } \hskip-2mm  &&x\s(k+1)=A\s x\s(k), \quad k \in \{0,1,\ldots\}
			\end{aligned}
		\end{align*}
		The data matrices $X\s^{-} \in \R{n \times N_s}$, $X\s^{+}\in \R{n \times N_s}$ which are defined as in \eqref{eq_xu_defn} are known to be available. Here, the matrices $S$, $T_i$ are defined as $S:=\smat{{X\s^-}^\top &{X\s^+}^\top}^\top \in \R{2n\times N\s}$, $T_i:=\smat{{X\t^{i-}}^\top&{X\t^{i+}}^\top}^\top \in \R{2n\times i}$.
		Here, it is assumed that the similar system data satisfies ${\rm rank}(X\s^-)=n$. This condition ensures that the underlying autonomous system is uniquely identified. It follows that $\dim(\ms)=n$, as ${\rm Rowsp}({X\s^+})\subseteq{\rm Rowsp}({X\s^-})$. 
		For this case too, the $\mathcal{H}_i$'s constructed from  $\ms$, $\mt_i$ as in \eqref{eq_Hi} inherits all the properties  discussed in \cref{sec_results}. Define $H_i\in\R{2n \times n}$ such that $\im{H_i}=\H_i$. Further partition $H_i=\begin{bsmallmatrix}{X_i^-}^\top  & {X_i^+}^\top\end{bsmallmatrix}^\top$
		such that  $X_i^- \in \R{n\times n}$, $X_i^+\in\R{n\times n}$.
		The unique system identified (say $A_i$) identified from $H_i$ is given by $ A_i = X_i^+ (X_i^-)^\dagger.$

		\section{Convergence of parameters}
		In this section, we study how the convergence properties of the $\H_i$'s translate to those of the system parameters $\ab{i}$ that are identified using the $\H_i$'s. We investigate the convergence of $\ab{i}$ computed from \eqref{eq_ABi} as $\H_i$'s gets computed from \eqref{eq_Hi} at time step $i$. More precisely, we investigate the dependencies between
		\begin{enumerate}[label=(\roman*)]
			\item $d(\H_i,\mt_\Omega)$ and {$\norm{\widehat{H}_i-\widehat{T}}_F$} for appropriately chosen basis $\widehat{H}_i$ and $\widehat{T}$, and
			\item $\norm{\widehat{H}_i-\widehat{T}}_F$ and $\norm{[A_i ~ B_i]-[A\t ~ B\t]}_F$ where $\ab{i}$ and $\ab{_T}$ are estimated from the matrices  $\widehat{H}_i$ and $\widehat{T}$ respectively.
		\end{enumerate}
		
		{The following example illustrates the fact that monotonicity in distance between subspaces doesn't necessarily translate to monotonicity in distance between systems represented by state space matrices.
			\begin{example}
				Consider the data sets 
				\begin{align*}
					D_0=\smat{1&0\\0&1\\1&1}, D_1=\smat{1&0\\0&1\\1.0149&1.0619}, D_2=\smat{1&0\\0&1\\1.0268&1.0633},
				\end{align*}
				which are collected from different systems with $m=n=1$. For $i=1,2,3$, let  $\mathcal{D}_i:=\im{D_i}$.  
				Here, 
				\begin{align*}d(\mathcal{D}_0,\mathcal{D}_1)=0.0257 > d(\mathcal{D}_0,\mathcal{D}_2)= 0.0252.\end{align*}
				It is easy to see that all the data sets $D_i$'s are informative for identification \citep{8960476}. For $i=0,1,2$, if we use $(a_i.b_i)$ to denote the corresponding system identified from $D_i$, then $(a_0,b_0)=(1,1)$, $(a_1,b_1)=(1.0149,1.0169)$ and $(a_2,b_2)=(1.0268,1.0633)$. Here,
				\begin{align*}
					\underbrace{\norm{[a_0~b_0]-[a_1,b_1]}_F}_{=0.0637} <  \underbrace{\norm{[a_0~b_0]-[a_2~b_2]}_F}_{=0.0687}.
				\end{align*}
				Though $\mathcal{D}_2$ is closer to $\mathcal{D}_0$ in subspace metric than $\mathcal{D}_1$, the same is not true in state space metric. \hfill $\square$
			\end{example}
			Recall that, if $\dim(\mathcal{A}\cap\mathcal{B})=i$, then $\theta_1=\theta_2=\cdots=\theta_i=0$ and the corresponding principal vectors satisfy $u_k=v_k$ and $u_k\in (\mathcal{A}\cap\mathcal{B})$ where $k=1,\ldots,i$ \citep{wedin2006angles}.
			In \cref{fig_planes1}, note that  $\dim(U\cap V)=1$ and hence  one principal angle $\theta_1=0$.}
		The following observation is used in our subsequent results. 
		\begin{lemma}\label{lem_zero_pangles}
			Assume that ${\dim (\mt_i )}=i ~\forall~i\in\{1,2,\ldots,\Lambda\}$. Then,
			\begin{enumerate}
				\item The pair $(\ms$, $\H_{i})$ has at least $\Lambda-i$ zero principal angles and at most $i$ non-zero principal angles.
				\item  The pair $(\mt_\Omega$, $\H_{i})$  has at least $i$ zero principal angles and at most $\Lambda-i$ non-zero principal angles. \label{statement_2}
			\end{enumerate}
		\end{lemma}\vspace*{-0.5\baselineskip}
		\begin{pf}
			Note that from \eqref{eq_Hi}, we get $\ms \cap \H_i = (\ms\cap\mt_i)+(\mt_i^\perp\cap \ms)$. Assuming ${\dim (\mt_i )}=i$, 	 it follows that, $\dim(\ms \cap \H_i)\geqslant \dim(\mt_i^\perp\cap \ms)=\Lambda-i$.
			As $\dim(\ms \cap \H_i)$ is at least $\Lambda-i$, we have at least $\Lambda-i$ principal angles to be zero for the pair $(\ms,\H_i)$. As, $\dim{\ms} = \dim{\H_{i}}=\Lambda$, $\Lambda$ principal angles are defined for the pair of subspaces $(\ms$, $\H_{i})$, and hence at most $i$ principal angles would be non-zero for the pair $(\ms$, $\H_{i})$.\\
			Statement (\ref{statement_2}) follows in a similar way. \hfill $\blacksquare$
		\end{pf}\vspace*{-0.5\baselineskip}
		\subsection{Bounding $\norm{\widehat{H}_i-\widehat{T}}_F$ above by a function of $d(\H_i,\mt_\Omega)$}
		Next, we propose a method, based on \cite[Theorem 1]{bjorck1973numerical}, for choosing a basis for $\H_i$ and $\mt_\Omega$.
		Let $Q\t$, $Q_i {~\in \R{M\times \Lambda}}$ be matrices containing orthonormal vectors that span $\mt_\Omega$, $\H_{i}$ respectively. Obtain the SVD of $Q\t^\top Q_i$ as $Q\t^\top Q_i=U\Sigma V^\top$, {where $U, \Sigma, V \in  \R{\Lambda\times \Lambda}$}. Define,
		\begin{equation}\label{eq_H_hat}
			\widehat{T}:=Q\t U, \quad \widehat{H}_i:=Q_iV.
		\end{equation}
		In the above equation, $\widehat{T}$, $\widehat{H}_i {~\in \R{M\times \Lambda}}$ contains principal vectors of the pair $(\mt_\Omega$, $\H_{i})$ and is such that $\im{\widehat{T}}=\mt_\Omega$, $\im{\widehat{H}_i}=\H_{i}$. 
		
		In a similar way, for the pair $(\ms$, $\H_{i})$,  we define matrices  $\widetilde{S}$, $\widetilde{H}_i {~\in \R{M\times \Lambda}}$ containing the corresponding principal vectors.
		The following lemma relates the distance between the matrices  $\widehat{T}$, $\widehat{H}_i$ with distance between the  subspaces $\mt_\Omega$, $\H_{i}$.
		\begin{lemma} \label{lem_bound_Hs}
			Let $\widehat{T}$, $\widehat{H}_i$ be defined as in  \eqref{eq_H_hat}. Assume that ${\dim (\mt_i )}=i ~\forall~i\in\{1,2,\ldots,\Lambda\}$. Then,
			\begin{equation*} 
				\norm{\widehat{H}_i-\widehat{T}}_F \leqslant 2\sqrt{\Lambda-i} \sin\left(\tfrac{d(\mt_\Omega,\H_i)}{2}\right).
			\end{equation*}
			Also, $\norm{\widetilde{H}_i-\widetilde{S}}_F \leqslant 2\sqrt{i} \sin\left(\frac{d(\ms,\H_i)}{2}\right)$ holds.
		\end{lemma}\vspace*{-0.5\baselineskip}
		\begin{pf}
			Let $h\t^{(k)}$, $h_i^{(k)} \in \R{M}$ be vectors used to denote the $k^\text{th}$ column of $\widehat{T}$, $\widehat{H}_i$ respectively. Then,  
			\begin{equation*}
				\norm{h_i^{(k)}-h\t^{(k)}}_2^2=4\sin^2\left(\tfrac{\theta_k}{2}\right) \left(\hspace{-0.2cm}\begin{tabular}{c}\text{See Proof of Lemma } $8$ \\[-4pt]in \citep{alsalti2024robust}.\end{tabular} \hspace{-0.2cm}\right) 
			\end{equation*}
			where $\theta_k$ is the  corresponding principal angle  between the principal vectors $h_i^{(k)}$, $h\t^{(k)}$. Hence, 
			
			\begin{align*}
				\norm{\widehat{H}_i-\widehat{T}}_F \hspace{-2.5mm}&=\hspace{-1mm}\sqrt{\sum_{k=1}^{\Lambda}\norm{h_i^{(k)}-h\t^{(k)}}_2^2}  = \sqrt{\sum_{k=1}^{\Lambda}4\sin^2\left(\tfrac{\theta_k}{2}\right)}\\ 
				&\hspace{-6mm}{=} 2\sqrt{\sum_{\mathclap{\; k=i+1}}^{\Lambda}\sin^2\left(\tfrac{\theta_k}{2}\right)} \left(\hspace{-0.2cm}\begin{tabular}{c} 
					By \cref{lem_zero_pangles}, at least $i$\\[-4pt]
					principal angles are zero.
				\end{tabular}\hspace{-0.2cm}\right)\\
				&\hspace{-6mm}\leqslant\hspace{-1mm} 2\sqrt{\hspace{-1mm}(\Lambda-i)\sin^2\hspace{-1mm}\left(\tfrac{\theta_{\Lambda}}{2}\right)} \left(\hspace{-0.2cm}\begin{tabular}{c} $\because \theta_1\hspace{-1mm}\leqslant\hspace{-1mm}\theta_2\leqslant\hspace{-1mm} \cdots\hspace{-1mm} \leqslant\hspace{-1mm}\theta_{\Lambda}$\\[-4pt] \text{\citep{bjorck1973numerical}}\end{tabular}\hspace{-0.2cm}\right)\\
				&\hspace{-6mm}=\hspace{-1mm}2\sqrt{\Lambda\hspace{-1mm}-\hspace{-1mm}i}\sin\hspace{-1mm}\left(\tfrac{d(\mt_\Omega,\H_i)}{2}\right) (\because\hspace{-1mm} d(\mt_\Omega,\H_i)\hspace{-1mm} =\hspace{-1mm}\theta_{\Lambda} \mbox{ from }\eqref{eq_dist_defn})
			\end{align*}
			The bound for {$\norm{\widetilde{H}_i-\widetilde{S}}_F$}  is derived in a similar way. \hfill $\blacksquare$
		\end{pf}\vspace*{-0.5\baselineskip}
		Note that, the bound for {$ \norm{\widehat{H}_i-\widehat{T}}_F$}  decreases with increase in $i$ because  $\H_i$ gets closer to $\mt_\Omega$ and so the matrix  $\widehat{H}_i$ also moves closer  to $\widehat{T}$. When $d(\mt_\Omega,\H_i)=0$, the right hand side becomes zero and $\widehat{H}_i$ exactly matches  $\widehat{T}$.  In contrast, the bound for { $\norm{\widetilde{H}_i-\widetilde{S}}_F$} increases as $i$ increases. This is because the subspace $\H_{i}$ relies more on the similar system data initially, and hence  $\widetilde{H}_i$ is also initially closer to $\widetilde{S}$. As $\H_{i}$ deviates from $\ms$, the matrix  $\widetilde{H}_i$ also moves farther from $\widetilde{S}$.
		\subsection{Bound on  relative deviation of $\ab{i}$ from $\ab{\t}$}
		Next, we aim to investigate the convergence of $\norm{\smat{ A_i ~B_i}-\smat{A\t ~B\t}}_F$. Let the matrices $\widehat{T}$, $\widehat{H}_i \in \R{M\times \Lambda}$ be partitioned as follows:
		$\widehat{T}=\smat{\widehat{X}\t^{-\top}& \widehat{U}\t^{-\top}& \widehat{X}\t^{+\top}}^\top$, $\widehat{H}_i=\smat{\widehat{X}_i^{-^\top}&\widehat{U}_i^{-\top}&\widehat{X}_i^{+\top}}^\top$ 
		where $\widehat{X}\t^-$, $\widehat{X}\t^+$, $\widehat{X}_i^-$, $\widehat{X}_i^+ \in \R{n\times \Lambda}$ and $\widehat{U}\t^-$, $\widehat{U}_i^-  \in \R{m\times \Lambda}$.
		\begin{align*}\hspace*{-10mm}
			\begin{array}{lll}
				\text{Let }&	\delta \widehat{X}_i^- := \widehat{X}_i^- - \widehat{X}\t^-, \qquad & \delta A_i:= A_i-A\t,\\
				&\delta \widehat{X}_i^+ := \widehat{X}_i^+-\widehat{X}\t^+, \qquad & \delta B_i:=B_i-B\t,\\
				&\delta \widehat{U}_i^- :=\widehat{U}_i^- -\widehat{U}\t^-. &
			\end{array}
		\end{align*}
		Then, by \cref{thm_ID}, $\widehat{X}_i^+=A_i\widehat{X}_i^-+B_i \widehat{U}_i^-$ holds and equivalently we can write
		\begin{equation}\label{eq_Xs1}
			\hspace{-2mm}\begin{smallmatrix}\widehat{X}\t^+ + \delta \widehat{X}_i^+\end{smallmatrix} = \smat{A\t +\delta A_i&  B\t+\delta B_i}\smat{\widehat{X}\t^- + \delta \widehat{X}_i^-\\ \widehat{U}\t^- + \delta \widehat{U}_i^-}.
		\end{equation}
		\begin{align}\label{notat_n_i}
			\begin{split}
				\text{\normalsize{Define:}}\qquad	n_1:={\norm{\begin{bsmallmatrix}\delta \widehat{X}_i^- \\ \delta \widehat{U}_i^- \end{bsmallmatrix}}_F,} n_2:&=\norm{\delta \widehat{X}_i^+}_F,\\
				\nu_1:=\norm{\begin{bsmallmatrix} \widehat{X}\t^- \\ \widehat{U}\t^- \end{bsmallmatrix}}_F, \nu_2:&=\norm{\widehat{X}_i^+}_F.
			\end{split}
		\end{align}
		\normalsize
		
		Denote $\kappa_2(\cdot)$ to be the $2$ condition number of a matrix.
		The following lemma bounds the magnitude of relative deviation between $\ab{i}$, $\ab{_T}$  caused as the data matrices $\begin{bsmallmatrix}\widehat{X}\t^-\\\widehat{U}\t^-\end{bsmallmatrix}$, $\widehat{X}\t^+$ deviate to $\begin{bsmallmatrix}\widehat{X}_i^-\\\widehat{U}_i^-\end{bsmallmatrix}$, $\widehat{X}_i^+$ respectively.
		\begin{lemma} \label{lem_delAB_bound1} 
			
			\begin{align} \label{eq_lem_delAB_bound1}
				\frac{\norm{\smat{\delta A_i ~ \delta B_i}}_F}{\norm{\smat{ A_i ~ B_i}}_F} \leqslant 
				\kappa_2\begin{bsmallmatrix}\widehat{X}\t^-\\ \widehat{U}\t^-\end{bsmallmatrix} \left(  \frac{n_1}{\nu_1} +\frac{n_2}{\nu_2} + 
				\frac{n_1 n_2}{\nu_1 \nu_2}   \right)
			\end{align}
		\end{lemma}\vspace*{-0.5\baselineskip}
		\begin{pf}	
			Applying the result of \cite[Theorem 2.3.8]{watkins2004fundamentals} row-wise by comparing $\widehat{X}\t^+\hspace{-1mm}=\hspace{-1mm}A\t\widehat{X}\t^-\hspace{-1mm}+\hspace{-1mm}B\t \widehat{U}\t^-$ and  \eqref{eq_Xs1}, and by relating it with the Frobenius norm we get the desired result. \hfill $\blacksquare$
		\end{pf}\vspace*{-0.5\baselineskip}
		
		Next, we combine \cref{lem_bound_Hs} and \cref{lem_delAB_bound1} to characterize the convergence of $\ab{i}$ in terms of the subspace distance $d(\H_i,\mt_\Omega)$. First, we state an intermediate lemma required to prove the main result in  \cref{thm_delAB1}. 
		\begin{lemma}\label{lem_blkmat}
			Consider a subspace $\mathcal{V} \subseteq \R{n}$, with $\dim(\mathcal{V})=q$. Let $Q_A$, $Q_B \in \R{n\times q}$ denote matrices containing orthonormal columns that span $\mathcal{V}$. Partition $Q_A=\begin{bsmallmatrix}
				A_1\\[0.1em] A_2
			\end{bsmallmatrix}$, $Q_B=\begin{bsmallmatrix}
				B_1\\[0.1em]B_2
			\end{bsmallmatrix}$ where $A_1$, $B_1 \in \R{q \times q}$ and $A_2$, $B_2 \in \R{(n-q) \times q}$. Then, $\norm{B_k}_F=\norm{A_k}_F$ for $k=\{1,2\}$.
			Also, if $A_1$, $B_1$ are invertible, then we have $\kappa_2(A_1)=\kappa_2(B_1)$.
		\end{lemma}\vspace*{-\baselineskip}
		\begin{pf}
			See Appendix. \hfill $\blacksquare$
		\end{pf}\vspace*{-\baselineskip}
		
		\subsection{Bound on System Identification Error as a function of $d(\H_i,\mt_\Omega)$}
		Next, we state the main result characterizing the error in identified parameters at time step $i$. Recall the definition of $h_i$ from \eqref{eq_Hibar}.

		\begin{theorem} \label{thm_delAB1}  Consider the notations as in  \eqref{eq_Xs1} and \eqref{notat_n_i}. Assume that ${\dim (\mt_i )}=i ~\forall~i\in\{1,2,\ldots,\Lambda\}$. Then the following inequality holds:
			\begin{equation} \label{eq_thm_delAB1}
				\hspace{-4mm}\frac{\norm{\smat{\delta A_i ~ \delta B_i}}_F}{\norm{\smat{ A_i ~ B_i}}_F} \leqslant  \gamma\kappa_2\begin{bsmallmatrix}\widehat{X}\t^-\\ \widehat{U}\t^-\end{bsmallmatrix} \hspace{-1mm}\left(  \frac{1}{\nu_1} +\hspace{-1mm}\frac{1}{\norm{h_b}_2} + 
				\frac{\gamma}{\nu_1 \norm{h_b}_2}   \right)
			\end{equation}
			where $\gamma :=2\sqrt{\Lambda-i} \sin\left(\frac{d(\mt_\Omega,\H_i)}{2}\right)$, $h_b \in \R{n}$ is defined to be the last $n$ components of a unit vector in direction of $h_1$.
		\end{theorem}\vspace*{-0.5\baselineskip}
		\begin{pf} First, we aim to bound the quantities $n_1$, $n_2$ in \eqref{notat_n_i}.
			Note that the matrices $\begin{bsmallmatrix}\delta \widehat{X}_i^- \\ \delta \widehat{U}_i^- \end{bsmallmatrix}$, $\delta \widehat{X}_i^+$ both are sub-matrices of $\widehat{T}-\widehat{H}_i$. Hence, by \cref{lem_bound_Hs},
			\begin{equation}\label{pf_bdd1}
				\{n_1,n_2\}   \hspace{-1mm}\leqslant \hspace{-1mm}  	\norm{\widehat{T}\hspace{-1mm}-\hspace{-1mm}\widehat{H}_i}_{F} \leqslant\hspace{-1mm} 2\sqrt{\Lambda-\hspace{-1mm}i}\sin\left(\hspace{-1mm}\tfrac{d(\mt_\Omega,\H_i)}{2}\hspace{-1mm}\right),
			\end{equation}
			where $n_1$, $n_2$ are defined as in \eqref{notat_n_i}.
			Next, we aim to bound $\frac{1}{\nu_2}$ in \eqref{eq_lem_delAB_bound1}. Recall  $h_1$  in \eqref{eq_Hibar}.  Let $\widehat{h}_1$ denote a unit vector in direction of $h_1$ and is partitioned as $\widehat{h}_1=\smat{h_a \\[0.1em] h_b}$ with $h_a \in {\R{\Lambda}}$ and $h_b \in \R{n}$. Note that, $\widehat{h}_1 \in \mt_i$, as $h_1 \in \mt_i$ holds $\forall ~ i  \in\{1,2,\ldots,\Lambda\}$.
			
			Let $Q_i\in {\R{M\times \Lambda}}$ be a matrix with  orthonormal columns such that $\im{Q_i}=\H_i$ and $\widehat{h}_i$ being one of its columns. Then, $\exists$ an orthogonal matrix $G {\in \R{\Lambda \times \Lambda}}$ such that $Q_i=\widehat{H}_iG$ (See Proof of Lemma $8$ in \citep{alsalti2024robust}).  Let $Q_i^b {\in \R{n \times \Lambda}}$ be a submatrix of $Q_i$ defined to be equal to last $n$ {rows} of $Q_i$. Then, by \cref{lem_blkmat}, $\norm{Q_i^b}_F=\norm{\widehat{X}_i^+}_F =\nu_2$. As $\norm{h_b}_2 \leqslant \norm{Q_i^b}_F~ \forall ~ i \in {\{1,2,\ldots,\Lambda\}}$, we have $\norm{h_b}_2 \leqslant  \nu_2$. Therefore,
			
			\begin{equation}\label{pf_bdd2}
				\frac{1}{\nu_2}\leqslant\frac{1}{\norm{h_b}_2} 
			\end{equation}\normalsize
			Applying \eqref{pf_bdd1} and \eqref{pf_bdd2} to the bound in \cref{lem_delAB_bound1}, we get the desired result. \hfill $\blacksquare$
		\end{pf}\vspace*{-0.5\baselineskip}
		Also by \cref{lem_blkmat}, since $\kappa_2\begin{bsmallmatrix}\widehat{X}\t^-\\ \widehat{U}\t^-\end{bsmallmatrix}$ and $\nu_1$ remains invariant with the choice of basis, the  bound in \eqref{eq_thm_delAB1} changes monotonically with $d(\mt_\Omega,\H_i)$. The above theorem leads us to the following conclusions
		\begin{enumerate}
			\item As $i$ increases, the uncertainty set containing the true system $\ab{_T}$ decreases.
			\item Since  $d(\H_i,\mt_\Omega)<d(\ms,\mt_\Omega)$, the $\H_i$'s  created as in \eqref{eq_Hi}, give a better uncertainty quantification than the initial similar system considered.
			\item The error in the parameter estimates tends to zero as $\H_i$ converges to $\mt_\Omega$.
		\end{enumerate}
		
		In the next result, we characterize the deviation in identified parameters with respect to $\ab{\s}$ at time step $i$.%
		
		Let the matrices $\widetilde{S}$ , $\widetilde{H}_i \in \R{M\times \Lambda}$ be partitioned as follows:
		$\widetilde{S}=\smat{\widetilde{X}\s^{-\top}& \widetilde{U}\s^{-\top}& \widetilde{X}\s^{+\top}}^\top$, $\widetilde{H}_i=\smat{\widetilde{X}_i^{-^\top}&\widetilde{U}_i^{-\top}&\widetilde{X}_i^{+\top}}^\top$  where $\widetilde{X}\s^-$, $\widetilde{X}\s^+$, $\widetilde{X}_i^-$, $\widetilde{X}_i^+ \in \R{n\times \Lambda}$ and $\widetilde{U}\s^-$, $\widetilde{U}_i^-  \in \R{m\times \Lambda}$.
		Let $\Delta A_i:= A_i-A\s$, $\Delta B_i:= B_i-B\s$ and $m_1:={\norm{\begin{bsmallmatrix} \widetilde{X}_i^- -\widetilde{X}\s^- \\  \widetilde{U}_i^- -\widetilde{U}\s^- \end{bsmallmatrix}}_F} $.
		
		\begin{theorem} \label{thm_delAB2}  Consider the notations as defined above. Assume that ${\dim (\mt_i )}=i ~\forall~i\in\{1,2,\ldots,\Lambda\}$. Then the following inequality holds:
			\begin{equation} \label{eq_thm_delAB2}
				\hspace{-3mm}\frac{\norm{\smat{\Delta A_i ~ \Delta B_i}}_F}{\norm{\smat{ A_i ~ B_i}}_F} \leqslant  \beta\kappa_2\begin{bsmallmatrix}\widehat{X}\s^-\\ \widehat{U}\s^-\end{bsmallmatrix} \hspace{-1mm}\left(  \frac{1}{m_1} +\hspace{-1mm}\frac{1}{\norm{h_b}_2} + 
				\frac{\beta}{m_1 \norm{h_b}_2}   \right)
			\end{equation}
			where $\beta :=2\sqrt{i} \sin\left(\frac{d(\ms,\H_i)}{2}\right)$, $h_b \in \R{n}$ is defined to be the last $n$ components of a unit vector in direction of $h_1$.
		\end{theorem}\vspace*{-0.5\baselineskip}
		\begin{pf}
			The above result can be obtained in  similar lines of \cref{thm_delAB1} and hence the proof is skipped. \hfill $\blacksquare$
		\end{pf}\vspace*{-0.5\baselineskip}
		Note that the above bound increases with $i$, as $\ab{i}$ move away from the similar system. 
		
		\begin{remark}
			The assumption $\dim(\mt_i)=i$ in the above results has been used for simplicity of the proofs. Our approach and all the guarantees provided are valid even when $\dim(\mt_i)=r<i$.  Under this case, the bounds in \cref{lem_bound_Hs} becomes  $\norm{\widehat{H}_i-\widehat{T}}_F \leqslant 2\sqrt{\Lambda-r} \sin\left(\frac{d(\mt_\Omega,\H_i)}{2}\right)$ and $\norm{\widetilde{H}_i-\widetilde{S}}_F \leqslant 2\sqrt{r} \sin\left(\frac{d(\ms,\H_i)}{2}\right)$. The bound given in \cref{thm_delAB1} also changes accordingly. For \cref{thm_H0star} if the assumption is not used, at some intermediate stage, we would encounter $\mt_i=\mt_{i+1}$ and hence $d(\mt_\Omega,\H_i)=d(\mt_\Omega,\H_{i+1})$. Hence, statement \ref{thm_H0star_p2} of \cref{thm_H0star} becomes  $d(\mt_\Omega, \mathcal{H}_i) \geqslant d(\mt_\Omega, \mathcal{H}_{j})$ for $i<j$.
		\end{remark}

		\section{Numerical implementation}
		In this section, we discuss the numerical aspects of implementing the proposed approach. A basis for $\mt_i^\perp \cap \ms$ is required at each time step for computing $\mathcal{H}_i$. 
		Denote $\ker{\cdot}$ to denote the kernel of a matrix $\cdot$ and ${\rm Basis}(\cdot)$ to denote a matrix whose columns form a basis for the subspace $\cdot$.  Denote $S^\perp$ as a matrix whose rows form a basis for the left-kernel of the matrix $S$.
		Then a basis (say $S_i$) for $\mt_i^\perp \cap \ms$ can be computed \citep{wonhamlinear} as $S_i = \mathrm{Basis}(\mt_i^\perp \cap \ms) = \mathrm{Basis}\left(\mathrm{Ker}\begin{bsmallmatrix}S^\perp \\ T_i^\top\end{bsmallmatrix}\right).$
		Then	$S^\perp$ is computed from the  SVD of $S$, and $S_i$ is computed from the SVD of the matrix $\begin{bsmallmatrix}		S^\perp \\ T_i^\top	\end{bsmallmatrix}$. 
		The computational procedure is summarized in Algorithm \ref{algo2}.

		\begin{algorithm}[h]
			\hrule\vspace{-0.5em}
			\caption{Incremental transfer identification scheme based on sequentially collected data} \label{algo2}
			\hrule
			\begin{enumerate}
				\item \textbf{Input:} $\{{u}\s(i),{x}\s(i)\}_{i=0}^{N_s}$ collected from the similar system.
				\item Construct the matrix $S$, as defined in \eqref{eq_S}, from the available similar system data.
				\item Compute the SVD of $S$ to obtain a basis for its left kernel (say $S^\perp$). \label{algo_p1}
				\item Initialize $T_i=0$  and assign $\mu=\Lambda$.\\
				{\bfseries while {$\mu \neq 0$}}
				\item Collect the current data and denote it by $h_i$.
				\item Update the matrix $T_i=\smat{T_{i-1}\quad h_i}$.
				\item Compute a basis for ${\mt_i^\perp}\cap{\ms}$ (say $S_i$) by performing SVD on $\begin{bsmallmatrix}
					S^\perp \\ T_i^\top \end{bsmallmatrix}$	and let $\mu=\ran{S_i}$. \label{algo_p2}
				\item Define $H_i=\smat{T_i \quad S_i }$.
				\item Partition the rows of $H_i$ as in \eqref{eq_part_Hi} to obtain the matrices $(X_i^-, U_i^-, X_i^+)$.
				\item Calculate  $\smat{A_i \quad B_i}=X_i^+\smat{X_i^-\\ U_i^-}^\dagger$. \label{algo_p3}\\
				\textbf{end while}
				\item \textbf{Output:} $(A_i,B_i)$ - An estimate of $\ab{\t}$ at time step $i$.
			\end{enumerate}
			\hrule
		\end{algorithm}

		Computing a basis for the kernel or the left-kernel of a matrix numerically typically involves performing a Singular Value Decomposition (SVD). Likewise, the computation of the Moore–Penrose pseudoinverse also necessitates an SVD step. For a matrix \( A \in \mathbb{R}^{m \times n} \) with \( m > n \), the computational complexity of the SVD is \(\mathcal{O}(mn^2)\) floating-point operations. Similarly, for multiplication of matrices of dimensions \( m \times n \) and \( n \times p \) requires \(\mathcal{O}(mnp)\) operations.
		See \citep{watkins2004fundamentals,golub2013matrix} for details. Computational costs involved in the steps \cref{algo2} are given in  \cref{tab_cost1}.
		\begin{table}[h] 
			\centering
			\caption{Computational Complexity}\label{tab_cost1}
			\begin{tabular}{|>{\centering\arraybackslash}m{0.6cm}|>{\centering\arraybackslash}m{4.25cm}|>{\centering\arraybackslash}m{2.6cm}|}  
				\hline
				\rule[-1ex]{0pt}{2.5ex}\vskip-0.5cm
				\textbf{Step No.}& \textbf{Steps involved} & \textbf{FLOPS
				} \\[-4pt]
				\hline
				\rule[-1ex]{0pt}{2.5ex}
				\ref{algo_p1}	  & Basis for the left-kernel of the matrix $S \in \R{M \times N_s}$ &  $\approx \mathcal{O}(M^2N_s)$\\[-3pt]
				\hline
				\rule[-1ex]{0pt}{2.5ex}
				\ref{algo_p2} &  Basis for the kernel of the matrix $\begin{bsmallmatrix}
					S^\perp \\ T_i^\top \end{bsmallmatrix}$ (where $S^\perp\in \R{M\times(M-\Lambda)}$, $T_i \in \R{M \times i}$) & $\approx \mathcal{O}((M-\Lambda+i)^2M)$ \\[-3pt]
				\hline
				\rule[-1ex]{0pt}{2.5ex}\vskip-0.3cm
				\ref{algo_p3}(i) & Pseudo-inverse of $\smat{X_i^-\\ U_i^-} \in \R{\Lambda \times \Lambda}$ & $\approx \mathcal{O}(\Lambda^3)$\\  
				\hline
				\rule[-1ex]{0pt}{2.5ex}\vskip-0.4cm
				\ref{algo_p3}(ii) & Matrix multiplication  & $ \approx \mathcal{O}(\Lambda^2n)$\\
				\hline
			\end{tabular}
		\end{table}

		\begin{remark}
			Note that $\dim (\mt_i)$ increases as we have new information from the actual system, and hence $\dim (\mt_i^\perp \cap \ms)$ decreases. When the  identification of the true system is complete, then  $\dim (\mt_i)=\Lambda$ and $\dim (\mt_i^\perp \cap \ms) = 0$ (due to \eqref{prop_eq}). This condition  has been used as the stopping criteria in Algorithm \ref{algo2}. At that point $ \mathcal{H}_i = \mt_i=\mt_\Omega$.
		\end{remark}

		\section{Illustrative Examples}
		
		In this section, we provide two numerical case studies to demonstrate the effectiveness of the proposed approach.
		\subsection{Transfer identification based on sequentially collected data}
		In this study, we apply the proposed method to  an LTI zone power model for a pressurized heavy water nuclear reactor (PHWR) from \citep{vaswani2022optimised}
		with state dimension of \( n = 56 \) and an input dimension of \( m = 15 \) as the true system, denoted by \( \ab{\t} \). Here, at any time instant $i$ our aim is to identify a suitable model for $\ab{\t}$ in online mode using data available till the current instant. For this case, $\Lambda=71$.
		
		We construct a perturbed version of the true system as the similar system \( \ab{\s} \). This perturbed system is defined as
		\(
		(A\t + \delta A_i, B\t + \delta B_i),
		\)
		where  the perturbation matrices \( \delta A_i \in \R{56\times 56}\) and \( \delta B_i \in \R{56\times 15}\) is uniformly sampled from the interval $(0,\sigma)$. Here, we consider two cases with $\sigma=0.05$ and $\sigma=0.1$. 
		
		A trajectory of length $300$ time steps is collected from the  system \( \ab{\s} \). The initial conditions and excitation inputs are randomly selected to ensure sufficient richness in the data. The true system \( \ab{\t} \) begins operation at time \( t = 0 \).
		
		At each time step \( i \), a data matrix \( T_i \in \R{127 \times i} \) is constructed based on the observed trajectory. Subsequently, the procedure outlined in Algorithm~\ref{algo2} is applied to estimate the system parameters, resulting in an identified system \( \ab{i} \).
		
		To evaluate the accuracy of the identification process, two metrics are computed at each stage:
		\begin{enumerate*}[label=(\roman*)]
			\item The distance \( d(\mathcal{H}_i, \mathcal{T}_\Omega) \), which quantifies the distance between the constructed subspace $\H_i$ and the subspace $\mt_\Omega$ corresponding to the true system.
			\item The Frobenius norm error \( \| [A_i \quad B_i] - [A\t \quad B\t] \|_F \), which measures the deviation of the estimated system matrices from the true matrices.
		\end{enumerate*}

		The evolution of these metrics over time is illustrated in Figure~\ref{fig_simulation}, providing insight into the convergence behavior.
		We have presented a comparison of two similar systems created by uniformly sampling the perturbation terms $\delta A_i$, $\delta B_i$ from the intervals $(0,0.05)$ and $(0,0.1)$. When $\ran{T_i}$ becomes equal to $71$, $d(\H_i,\mt_\Omega)$  becomes zero and hence $\H_i=\mt_\Omega$ and $\ab{i}=\ab{\t}$.

		\begin{figure}
			\centering
			\includegraphics[width=0.8\linewidth,height=0.45\linewidth]{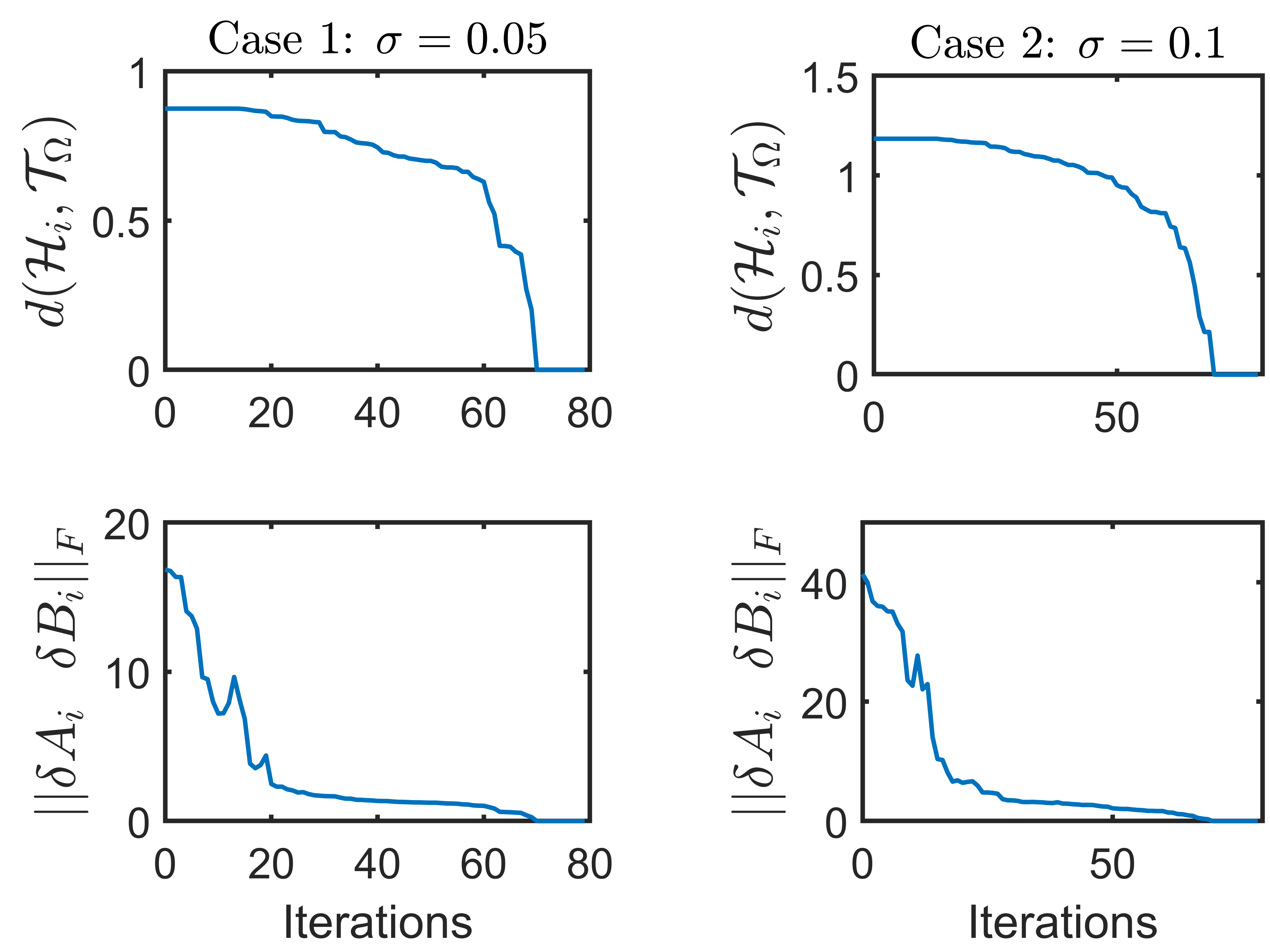}
			\caption{Variation of $d(\H_i,\mt_\Omega)$ and $\norm{[A_i \quad B_i]-[A\t \quad B\t]}_F$}
			\label{fig_simulation}
		\end{figure}
		\subsection{Pole placement using certainty equivalence approach with batch data}
		Here, we consider an open unstable system given in \citep{dean2020sample} 
		\begin{align*}
			A\t=\smat{1.01 &\quad0.01 &\quad0\\0.01 &\quad1.01&\quad 0.01\\0 &\quad0.01 &\quad1.01}, B\t= I_3.
		\end{align*}
		Our goal is to find a state feedback controller to place the closed-loop poles of the above system  at $\{-0.5,0.5,0.75\}$.
		Also, consider the similar system
		\begin{align*}
			A\s\hspace{-1mm}=\hspace{-1mm}\smat{0.0560   &-0.2909  &  0.1998\\[3pt]
				-1.0053   & 0.2756  & -0.4477\\[3pt]
				0.1049   & 0.3415  &  1.5790}, \hspace{-0.5mm}B\s\hspace{-1mm}=\hspace{-1mm}\smat{0.6828 & \:  0.1730   & \:0.1366\\[3pt]
				-0.0453   &\: 1.6228   &\: 0.0700\\[3pt]
				-0.1402   &\: 0.1571   &\: 0.6139}.
		\end{align*}
		The above matrices are obtained by perturbing $\ab{\t}$ entrywise with normal random variables of mean $0$ and standard deviation $0.5$.
		Here,\linebreak $\norm{[A\s ~ B\s]-[A\t ~ B\t]}_F=2.0006$. It is assumed that we have a data set from the similar system satisfying $\ran{S}=\Lambda=6$ and hence $\ab{\s}$ is known. Whereas, $\ab{\t}$ is unknown, but we have the following data collected from $\ab{\t}$
		\begin{align*}
			\{u_k\}_{k=0}^{3}&=\begin{Bsmallmatrix}\smat{1\\1\\1},\smat{-1\\1\\1},\smat{1\\-1\\1},\smat{1\\-1\\-1}\end{Bsmallmatrix},\\
			\{x_k\}_{k=0}^{4}&=\hspace{-1mm}\begin{Bsmallmatrix}\hspace{-1mm}\smat{1\\[2pt]-1\\[2pt]-1},\hspace{-0.5mm}\smat{2\\[2pt]-0.01\\[2pt]-0.02}, \hspace{-0.5mm}\smat{1.0199\\[2pt]1.0097\\[2pt]-1.0203},\hspace{-0.5mm}\smat{2.0402\\[2pt]0.0198\\[2pt]-0.0204},\hspace{-0.5mm} \smat{3.0608\\[2pt]-0.9598\\[2pt]-1.0204}\hspace*{-1mm}\end{Bsmallmatrix}.
		\end{align*}

		From the available data, a data matrix $T_3$ is constructed, and here $\ran{T_3}=3<\Lambda$. Subsequently, we construct the subspace $\H_3$ according to \eqref{eq_Hi} and identify the system
		
		\begin{align*}
			A_3\hspace{-1mm}=\hspace{-1mm}\smat{ 0.7857   & \: 0.1184  &\: -0.5471\\[3pt]
				-0.3831   &\: 1.0674  &\: -1.0759\\[3pt]
				0.0984   & \: 0.3010   &\: 1.5753}, B_3\hspace{-1mm}=\hspace{-1mm}\smat{1.1150  &\:  0.1123  & \:-0.4416\\[3pt]
				0.0717   & \:1.0648 &\:  -0.7718\\[3pt]
				0.2814  & \: 0.2876 & \:  1.1888},
		\end{align*}
		and $\norm{[A_3 ~ B_3]-[A\t ~ B\t]}_F=1.7708$. Now, we consider $\ab{3}$ as an approximate model/estimate of $\ab{\t}$ and we aim to solve the desired pole-placement problem disregarding the modeling error. This is called the certainty equivalence approach in the literature \citep{dorfler2023certainty,mania2019certainty}. If $\ab{\t}$ is known, one could easily find a state feedback controller to place the poles of the closed loop system precisely at $\{-0.5,0.5,0.75\}$. Since the true system is unknown, the best one can do is to find a state feedback controller that places the closed-loop poles in a neighborhood of the desired pole locations.
		
		Using MATLAB's \texttt{place} command, we obtain the controller 
		$K_3=\smat{ 0.3218  &\quad -0.0667 &\quad  -0.1242\\[2pt]
			-0.3204   &\quad 1.4221 &\quad  -0.4062\\[2pt]
			0.0841&\quad   -0.0751  &\quad  0.8219}$, which places the poles of $\ab{3}$ at $\{-0.5,0.5,0.75\}$. The closed loop poles of $(A\t-B\t K_3)$ are $\{-0.4844,0.2565,0.6921\}$ and $\norm{K_3-K\t}_F= 0.8228$. The resulting closed-loop system is stable, with poles located close to our desired position.
		
		Suppose one uses the same certainty equivalence approach by considering $\ab{\s}$ as our estimate of $\ab{\t}$. In that case, we obtain a controller $K\s=\smat{-0.5301& \quad  -0.6005& \quad   0.0865\\[2pt]
			-0.6435   & \quad0.4480&   \quad-0.3363\\[2pt]
			0.2145    &\quad0.3045 & \quad  1.4562}$  which places the poles of $\ab{\s}$ at $\{-0.5,0.5,0.75\}$. However, the closed-loop system $(A\t-B\t K\s)$ has poles at $\{-0.3319,0.1477,1.8401\}$ and hence it is unstable.
		\begin{figure}
			\centering
			\includegraphics[width=0.6\linewidth,height=0.5\linewidth]{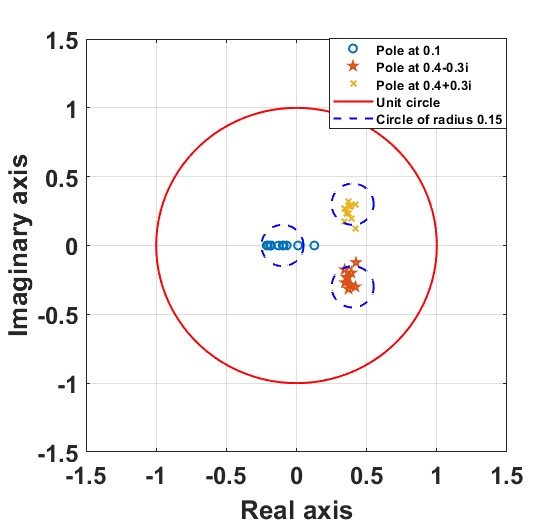}
			\caption{Closed loop poles of $\ab{\t}$}
			\label{fig_pole_placement}
		\end{figure}
		
		\cref{fig_pole_placement} shows the closed-loop poles of $\ab{\t}$ for $10$ different realizations obtained using the certainty equivalence approach, with $\{-0.1,\,0.4 \pm 0.3\iota\}$ as the desired pole locations. For each case, a data trajectory satisfying $\ran{T_i} = 3$ is collected from $\ab{\t}$. The matrices $\ab{\s}$ are generated by perturbing $\ab{\t}$ entrywise with normal random variables of mean $0$ and standard deviation $0.125$.
		From \cref{fig_pole_placement}, we can infer that the closed-loop poles are quite close to the desired locations.
		
		\section{Conclusion and future works}\label{sec_con}
		We investigated the problem of transfer identification for LTI systems from a geometric perspective. This approach of characterizing systems led us to a simple closed-form expression to combine the data from the similar system and the true system. In addition, several attributes of the constructed subspaces $\H_i$ have been analyzed. We have also characterized a bound that captures the error in the identified system matrices $\ab{i}$ at any time instant $i$. To illustrate the effectiveness of the proposed approach, two numerical case studies have been conducted: one on a high-dimensional system and another focused on pole placement. Extending these deterministic results to the noisy data setting remains a significant challenge and is the focus of ongoing research.

		\begin{ack}                               
			The authors thank Vatsal Kedia for useful discussions and valuable inputs regarding this work.
		\end{ack}
		\scriptsize
		\bibliographystyle{plainnat}        
		\bibliography{bibfile_id}     
		\normalsize
		
		\begin{appendix} 
			\section{Appendix}\label{appendicx}
			\subsection{Proof of \cref{lem_dist_eq}}\label{ap_pf1}
			Let $Q_\D=\smat{d_1~d_2\cdots d_k}$. It is easy to see that $\smat{Q_\B ~Q_\D}$ forms an orthonormal basis for $\C$ as $d_i\perp\B$ for $i=1,2,\ldots,k$.
			Next consider,
			\begin{align*}
				\begin{bsmallmatrix} Q_\mathcal{B}^\top \\ Q_\D^\top \end{bsmallmatrix}\begin{smallmatrix}Q_\mathcal{A}Q_\mathcal{A}^\top\end{smallmatrix} \begin{bsmallmatrix}
					Q_\mathcal{B} ~~ Q_\D
				\end{bsmallmatrix} &= \begin{bsmallmatrix}
					Q_\mathcal{B}^\top Q_\mathcal{A}Q_\mathcal{A}^\top Q_\mathcal{B} \quad Q_\mathcal{B}^\top  Q_\mathcal{A}Q_\mathcal{A}^\top Q_\D \\   Q_\D^\top Q_\mathcal{A}Q_\mathcal{A}^\top Q_\mathcal{B} \quad Q_\D^\top Q_\mathcal{A}Q_\mathcal{A}^\top Q_\D
				\end{bsmallmatrix} \nonumber \\ &=\begin{bsmallmatrix}
					Q_\mathcal{B}^\top Q_\mathcal{A}Q_\mathcal{A}^\top Q_\mathcal{B} & 0_{q\times k}\\ 0_{k \times q} & {\rm I}_k
				\end{bsmallmatrix},
			\end{align*}
			where $q=\dim(\B)$.
			Note that $Q_\mathcal{A}Q_\mathcal{A}^\top$ is the orthogonal projection onto $\mathcal{A}$, hence $Q_\mathcal{A}Q_\mathcal{A}^\top d_i = d_i$ as $d_i\in \A$. As $d_i$'s are orthogonal to $\mathcal{B}$, $Q_\mathcal{B}^\top  Q_\mathcal{A}Q_\mathcal{A}^\top d_i =  Q_\mathcal{B}^\top d_i =0$ holds. Similarly, $d_i^\top Q_\mathcal{A}Q_\mathcal{A}^\top d_i =d_i^\top d_i = 1$.
			
			As singular values of a matrix of the form  $Q_\mathcal{A}^\top Q_\mathcal{B}$ are cosines of principal angles \cite[Theorem 1]{bjorck1973numerical}, $1$ corresponds to the maximum singular value. Hence, we conclude that $ \sigma_{\min} (Q_\mathcal{A}^\top Q_\mathcal{B})=\sigma_{\min}(Q_\mathcal{A}^\top[Q_\mathcal{B} \quad Q_\D])$. And since the minimum singular value determines the maximum principal angle between the subspaces, equivalently, we have $d( \mathcal{A}, \mathcal{B}) = d( \mathcal{A}, \mathcal{C})$.
			
			Next to prove statement \ref{lem1_p2}, consider $\D = \mathcal{A}\cap \mathcal{B}^\perp$. As $\B$ is not partially orthogonal to $\A$, in the similar lines of obtaining \eqref{prop_eq}, here we have $\dim(\D)=\dim(\mathcal{A}\cap \mathcal{B}^\perp)=\dim(\A)-\dim(\B)$. As $\C=\B\oplus\D$,  $\dim(\C)\hspace{-1mm}=\hspace{-1mm}\dim(\B)+\dim(\D)\hspace{-1mm}=\hspace{-1mm}\dim(\A)$. 
			\subsection{Proof of \cref{lem_dist_eq2}}\label{ap_pf2}
			Let $\dim(\A\cap\C)=\gamma$, $\dim(\A\cap\B)=\beta$. Then, $\exists$ matrices $Q_{\A}$, $Q_{\B}$, $Q_{\C}$ with orthonormal columns that spans $\A$, $\B$, $\C$ respectively and the symmetric matrix $Q_{\C}^\top Q_{\A}Q_{\A}^\top Q_{\C}$ can be partitioned as 
			\begin{align}\label{eq_blkmat1}
				\begin{smallmatrix}Q_{\C}^\top Q_{\A}Q_{\A}^\top Q_{\C}\end{smallmatrix}=\begin{bsmallmatrix} \mathcal{N} &  &  \\ & {\rm I}_{\beta} &\\ & & {\rm I}_{\gamma-\beta} \end{bsmallmatrix},
			\end{align}
			with $\begin{smallmatrix}Q_{\B}^\top Q_{\A}Q_{\A}^\top Q_{\B}\end{smallmatrix}=\begin{bsmallmatrix} \mathcal{N} &   \\ & {\rm I}_{\beta}  \end{bsmallmatrix}$. In the above block matrix, $\dim(\A\cap\C)=\gamma$, we have $\gamma$ principal angles to be zero for the pair $(\A,\C)$ and hence $Q_{\A}^\top Q_{\C}$ has $\gamma$ singular values exactly equal to $1$. Hence, \eqref{eq_blkmat1} has an identity matrix of size $\gamma$. A similar argument can be used for the pair $(\A,\B)$ and one can conclude that $Q_{\B}^\top Q_{\A}Q_{\A}^\top Q_{\B}$ would have an identity matrix of size $\beta$.
			The symmetric matrix $\mathcal{N}$ in \eqref{eq_blkmat1} has eigenvalues not equal to $1$. 
			\newpage
			Then, the matrix $Q_{\C}$ can be partitioned as
			$Q_{\C}=\smat{Q_{\B} ~ d_1 ~ d_2 \cdots d_{\gamma-\beta}}$. Clearly, each $d_i \perp \B$ for $i=1,2,\ldots,(\gamma-\beta)$. Also as $d_i^\top Q_{\A}Q_{\A}^\top d_i =1$, we have $d_i \in \A$ for $i=1,2,\ldots,(\gamma-\beta)$. As $d_i$'s satisfies $d_i \in \A \cap \B^\perp$,  $\D=\langle d_1,d_2, \ldots, d_{\gamma-\beta} \rangle$ is the required subspace and is of dimension $\dim(\A\cap\C)-\dim(\A\cap\B)$.
			\subsection{Proof of \cref{lem_blkmat}}\label{ap_pf3}
			From the  Proof of Lemma $8$ in \citep{alsalti2024robust},  $\exists$ an orthogonal matrix $G$ such that  $Q_B=Q_A G$. Hence, we have $B_1=A_1G$, $B_2=A_2G$. Hence, $\norm{B_k}_F=\norm{A_k}_F$ for $k=\{1,2\}$ follows. 
			
			If, $A_1$, $B_1$ are invertible, then $\norm{A_1^{-1}}_2\hspace{-1mm}=\hspace{-1mm}\norm{G^\top B_1^{-1}}_2\hspace{-1mm}=\hspace{-1mm}\norm{B_1^{-1}}_2$. Hence, $\kappa_2(A)\hspace{-1mm}=\hspace{-1mm}\norm{A_{1}}_2\norm{A_1^{-1}}_2\hspace{-1mm}=\hspace{-1mm}\norm{B_{1}}_2\norm{B_1^{-1}}_2\hspace{-1mm}=\hspace{-1mm}\kappa_2(B_1)$.
			
		\end{appendix}	
	\end{document}